\documentclass[natbib,hyperref,doi,envcountsame]{svjour3}    
\smartqed  
\usepackage{amsmath, amssymb, mathrsfs, relsize}
\usepackage{mathtools}
\usepackage{enum}
\usepackage[english]{babel}

\usepackage{marginnote}
\usepackage{tikz}

\usepackage[algoruled, vlined]{algorithm2e}
\DontPrintSemicolon

\usepackage{subfigure}
\usetikzlibrary{arrows,topaths}
\usetikzlibrary{shapes}
\usetikzlibrary{fit,matrix}
\usetikzlibrary{automata, positioning}

\spnewtheorem{definition}[theorem]{Definition}{\bfseries}{\rmfamily}

\def\qed{\ifmmode\squareforqed\else{\unskip\nobreak\hfil
\penalty50\hskip1em\null\nobreak\hfil\squareforqed
\parfillskip=0pt\finalhyphendemerits=0\endgraf}\fi}

\newcommand{\NLOGSPACE}{\textsc{NLogSpace}}

\newcommand*{\pset}[1]{\mathcal{P}(#1)}
\newcommand*{\calG}{\mathcal{G}}
\newcommand*{\calA}{\mathcal{A}}
\newcommand*{\calS}{\mathcal{S}}
\newcommand*{\calP}{\mathcal{P}}
\newcommand*{\calK}{\mathcal{K}}
\newcommand*{\calM}{\mathcal{M}}
\newcommand*{\Hist}{\mathrm{Hist}}
\newcommand*{\rank}{\mathrm{rank}}

\tikzset{accepting by double} 
\tikzstyle{every state}=[minimum size=1em, inner sep=1pt, font=\scriptsize]

\newcommand*{\gob}{\textrm{\tiny $\bullet$}}
\newcommand*{\goc}{\raisebox{-1.2pt}{\textrm{\tiny $\bullet$}}\kern-3.7pt\raisebox{2pt}{\textrm{\tiny $\bullet$}}}
\newcommand*{\gout}{\raisebox{.5pt}{\kern-1pt\textrm{\tiny $\times$}\kern-1pt}}
\newcommand*{\gow}{\textrm{\tiny $\circ$}}
\newcommand*{\gocbw}{\raisebox{-1.2pt}{\textrm{\tiny $\bullet$}}\kern-3.7pt\raisebox{2pt}{\textrm{\tiny $\circ$}}}
\newcommand*{\gocwb}{\raisebox{-1.2pt}{\textrm{\tiny $\circ$}}\kern-3.7pt\raisebox{2pt}{\textrm{\tiny $\bullet$}}}

\newcommand*{\gocw}{\raisebox{-1.2pt}{\textrm{\tiny $\circ$}}\kern-3.7pt\raisebox{2pt}{\textrm{\tiny $\circ$}}}

\tikzset{
  node/.style = {draw, minimum size=.2em, inner sep=0pt, font=\scriptsize
    font=\sffamily}
}

\tikzset{every picture/.style={node distance=5em,->,>=latex,shorten >=1pt,auto, font=\scriptsize}} 
\tikzset{initial text=, initial distance=1em} 
\tikzset{accepting by double} 
\tikzstyle{every state}=[minimum size=1.7em, inner sep=2pt, font=\scriptsize]

\tikzset{
state1P/.style={
		rectangle,
		fill=gray!50,
		rounded corners,
       draw=black, 
       minimum height=2em,
       inner sep=5pt,
       text centered,
       }
}

\tikzset{
state3P/.style={
		rectangle split,
       rectangle split horizontal,
       rectangle split parts=3,
       rectangle split part fill={none,gray!50,none},
       rounded corners,
       draw=black, 
       minimum height=1em,
       inner sep=3pt,
       text centered,
       }
}

\tikzset{
state5P/.style={
		rectangle split,
       rectangle split horizontal,
       rectangle split parts=5,
       rectangle split part fill={none,none,gray!50,none,none},
       rounded corners,
       draw=black, thick,
       minimum height=2em,
       inner sep=5pt,
       text centered,
       }
}

\journalname{Acta Informatica}
\begin{document}

\title{Hierarchical Information and the Synthesis of Distributed Strategies}

\titlerunning{Hierarchical Information}        

\author{Dietmar Berwanger \and
        Anup Basil Mathew \and 
        Marie van den Bogaard %
}

\authorrunning{D. Berwanger, A.B. Mathew, M. van den Bogaard} 

\institute{D. Berwanger, A. B. Mathew, and M. van den Bogaard \at
              LSV, CNRS \& Universit\'e Paris-Saclay, France \\
              \email{dwb@lsv.fr} 
}

\date{Received: date / Accepted: date}

\maketitle

\begin{abstract}
Infinite games with imperfect information are known to be undecidable 
unless the information flow is severely restricted. 
One fundamental
decidable case occurs when there is a total ordering among players, such
that each player has access to all the information that the following ones receive.  
In this paper we consider variations of this hierarchy principle for
synchronous games with perfect recall, and identify new decidable
classes for which the distributed synthesis problem is solvable with
finite-state strategies. 
In particular, we show that decidability is maintained when the
information hierarchy may change along the play, or when transient phases
without hierarchical information are allowed. 
Finally, we interpret our result in terms of distributed system 
architectures. 

\end{abstract}

\section{Introduction}

To realise systems that are correct by design is a persistent 
ambition in computing science. 
The stake is particularly high for systems that interact with
an unpredictable environment over indeterminate time. 
Pioneering results in the area of synthesis, due 
to~\cite{BuechiLandweber69}, and~\cite{Rabin72}, 
show that the task can be 
automatised for the case of monolithic designs with correctness 
conditions specified by automata over infinite objects --- words or trees
representing computations. 
A most natural framework for 
representing and solving the problem is in terms of infinite games with
perfect information over finite graphs, as described 
by~\cite{PnueliRos90} or by~\cite{Thomas95}.

For distributed systems in which multiple components interact with the objective of 
satisfying a global specification, the game-theoretical formulation
of the synthesis problem leads to games with imperfect information and
to the question of whether there exists a winning strategy that can be 
distributed among the multiple players.
Unfortunately, such games are much less amenable to automated
solutions: as pointed out by~\cite{PetersonRei79}, 
it is generally undecidable whether a solution ---\,that is, a distributed
winning strategy\,--- 
exists for a finitely presented game for two players against Nature (or the
environment); furthermore,~\cite{Janin07} showed that,
even if a solution exists, it may not be
implementable by a finite-state device. As there is no
hope for solving the distributed synthesis problem uniformly, 
it remains to look out for classes that allow for an algorithmic
treatment. For surveys on results in this direction, see, e.g., the 
article of~\cite{GastinSZ09} or the theses 
of~\cite{Schewe08} and of~\cite{Puchala13}.

One fundamental case in which the distributed synthesis problem becomes
decidable is that of hierarchical systems: these
correspond to games where there exists a total order among the players
such that, informally speaking, each player has access to the information
received by the players that come later in the order. 
For such games,~\cite{PetersonRei79} 
showed that it is decidable\,---\,although, 
with nonelementary complexity\,---\,whether distributed winning
strategies exist and
if so, finite-state winning strategies can be effectively
synthesised. 
The result was extended by~\cite{PnueliRos90}
to the framework of distributed system architectures 
with linear-time specifications 
over path-shaped communication architectures where
information can flow only in one direction. 
Later,~\cite{KupfermanVar01} developed an 
automata-theoretic approach
that allows to extend the decidability
result from linear-time to branching-time specifications, and also
relaxes some of the syntactic restrictions imposed by the
fixed-architecture setting of Pnueli and Rosner. 
Finally,~\cite{FinkbeinerSch05} gave an effective
characterisation of communication architectures 
on which distributed synthesis is decidable. The criterion requires
absence of information forks, which implies a hierarchical order in
which processes, or players, have access to the 
observations emitted by the environment.

The setting of games is more liberal than that of architectures with
fixed communication channels. 
A rather general, though non-effective condition for games to 
admit finite-state distributed winning strategies is given
in~\cite{BKP11}, based on epistemic models representing the knowledge
acquired by players in a game with perfect recall. This condition
suggests that, beyond the fork-free architecture classification there
may be further natural classes of games for which the distributed synthesis
problem is decidable.

In this paper, we study a relaxation of the
hierarchical information pattern underlying the basic decidability
results on games with imperfect information and 
distributed system architectures. 
Firstly, we extend the assumption of hierarchical observation, that is,
\emph{positional} information, by incorporating perfect recall. 
Rather than requiring that a player \emph{observes} the signal
received by a less-informed player, we require that he can \emph{infer}
it from his observation of the play history. 
It can easily be seen that this gives rise to a decidable class, 
and it is likely that previous
authors had a perfect-recall interpretation in mind when describing
hierarchical systems, even if the formal definitions in the
relevant literature generally refer to observations.

Secondly, we investigate the case when the hierarchical information
order is not fixed,  
but may change dynamically along the play. 
This allows to model situations where the schedule of the 
interaction allows a less-informed player
to become more informed than others, or where the players may
coordinate on designating one to receive certain signals, and thus
become more informed than others. 
We show that this condition of dynamic hierarchical observation also
leads to a decidable class of the distributed synthesis problem.

As a third extension, we consider the case where the condition of
hierarchical information (based on perfect recall) is intermittent. That
is, along every play, it occurs infinitely often that the information
sets of players are totally ordered; nevertheless, there may be
histories at which incomparable information sets arise, as it is
otherwise typical of information forks. 
We show that, at least for
the case of winning conditions over attributes observable by all
players, this condition of recurring hierarchical observation is already
sufficient to ensure decidability of the synthesis problem, and that
finite-state winning strategies exist for all solvable instances.

For all the three conditions of hierarchical information, it is decidable with
relatively low complexity whether they hold for a given game. 
However, the complexity of solving a game is nonelementary in all
cases, as they are more general than the condition of hierarchical
observation, for which it was shown by~\cite{PetersonRei79}
that there exists no elementary solution procedure.

The last part of the paper presents an interpretation of the game-theoretic results in terms of 
distributed reactive sytems. Towards this, we extend the framework introduced 
by~\cite{PnueliRos90} which features hard-wired communciation graphs, to a model where the 
global actions are transduced into signals for the individual processes 
by a deterministic finite-state monitor. 
In this framework of monitored architectures, 
we identify classes that correspond to games with hierarchical information and 
therefore admit an effective solution of the distributed synthesis problem.

\section{Preliminaries}

\subsection{Games on graphs}

We use the standard model of concurrent games with imperfect
information, following the notation from~\cite{BKP11}. 
There is a set $N = \{1, \dots, n\}$ of  players 
and a distinguished agent called Nature. 
We refer to a list of elements
$x=(x^i)_{i \in N}$, one for each player, as a \emph{profile}.  
For each player~$i$, we fix a set $A^i$ of \emph{actions} and a set
$B^i$ of \emph{observations}; these are finite sets.

A~\emph{game graph} $G = (V, E, (\beta^i)_{i \in N})$ 
consists of a set $V$ of nodes called \emph{positions}, 
an edge relation $E \subseteq V \times A \times V$ 
representing simultaneous \emph{moves} 
labelled by action profiles, 
and a profile of \emph{observation} functions $\beta^i: V \to B^i$
that label every position with an observation for each player.
We assume that the graph has no dead ends, that is, for every
position~$v \in V$ and every
action profile~$a \in A$, there exists an outgoing move $(v, a, w) \in E$, and we 
denote the set of successors of a position~$v$
by $vEA := \{w~|(v,a,w) \text{ for some } a \in A\}$.

Plays start at a designated initial position $v_0 \in V$ 
and proceed in rounds. 
In~a round at position $v$, 
each player~$i$ chooses simultaneously
and independently an action $a^i \in A^i$, then  
Nature chooses a successor position $v'$ reachable 
along a move $(v, a, v') \in E$.  
Now, each player~$i$ receives the observation $\beta^i( v')$,
and the play continues from position $v'$.
Thus, a \emph{play} 
is an infinite sequence $\pi = v_0, v_1, v_2, \dots$ 
of positions, such that for all $\ell \ge 0$, there exists a move
 $(v_{\ell}, a, v_{\ell + 1}) \in E$.
A~\emph{history} is a nonempty prefix $\pi = v_0, v_1, \dots, v_{\ell}$ 
of a play; we refer to~$\ell$ as the
\emph{length} of the history, and we denote by 
$\Hist( G )$ the set of all histories in the game graph $G$. 
The observation function
 extends from positions 
to histories\footnote{Note that
we discard the observation at the initial position; this is technically
convenient and does not restrict the model.} and plays
as $\beta^i( \pi ) =\beta^i (v_1) \beta^i (v_2) \dots$, and we
write
$\Hist^i( G ) := \{ \beta^i( \pi )~|~ \pi \in \Hist( G ) \}$ for
the set of \emph{observation histories} of player~$i$.
We say that 
two histories $\pi, \pi'$ are \emph{indistinguishable} to player~$i$, and write 
$\pi \sim^i \pi'$, if they yield the same observation $\beta^i(\pi) = \beta^i(\pi')$. 
This is an equivalence relation, and its classes are 
called \emph{information sets}. 
The information set of player~$i$ at history~$\pi$ is 
$P^i( \pi ) := \{ \pi' \in \Hist( G)~|~ \pi' \sim^i \pi \}$. 
In terms of the taxonomy for distributed systems by~\cite{HalpernVar89}, 
our model is \emph{synchronous} and of \emph{perfect recall}. 

A \emph{strategy} for player~$i$ is a mapping $s^i: V^* \to A^i$
from histories to actions 
that is \emph{information-consistent} in the sense 
that  $s^i ( \pi ) = s^i( \pi')$, 
for any pair ~$\pi \sim^i \pi'$ 
of indistinguishable histories.
We denote the set of all strategies of player~$i$ 
by $S^i$ and the set of all strategy profiles by $S$.
A history or play $\pi = v_0, v_1, \dots$ 
\emph{follows} the strategy $s^i \in S^i$ if, for every  $\ell > 0$,
we have $( v_\ell, a, v_{\ell+1}) \in E$
for some action profile~$a$ where
$a^i = s^i( v_0, v_1, \dots, v_\ell)$.
The play $\pi$ follows a strategy
profile~$s \in S$ if it follows all component strategies~$s^i$.
The set of  possible \emph{outcomes} of a strategy profile~$s$ 
is the set of plays that follow~$s$.

A \emph{winning condition} over a game graph $G$ is a set 
$W \subseteq V^\omega$ of plays. 
A~\emph{distributed game} $\calG = (G, W)$ 
is described by a game graph and a winning condition. 
We say that a play $\pi$ is winning in~$\calG$ if $\pi \in
W$. 
A strategy profile $s$ is winning in~$\calG$ 
if all its possible outcomes are so. In this case, we refer to $s$ as
a \emph{distributed winning strategy}.
We generally assume that the game graph is finite 
or finitely presented.
The general \emph{distributed synthesis problem} asks whether for a 
given game graph and a fixed or given winning condition, 
there exists a distributed winning strategy. 

\subsection{Automata}

Our focus is on finitely-represented games, where the game graphs are
finite and the winning conditions described by finite-state automata.
Specifically, winning conditions are given by a colouring function
$\gamma: V \to C$ 
and an $\omega$-regular set $W \subseteq C^\omega$ 
describing the set of plays $v_0, v_1, \dots$ with $\gamma( v_0 ), \gamma( v_1), \dots  \in W$.
In certain cases, we assume that the colouring is \emph{observable} 
to each player~$i$, 
that is, $\beta^i( v ) \neq
\beta^i( v')$ whenever $\gamma( v ) \neq \gamma( v')$.
For general background on automata for games,  
we refer to the handbook by~\cite{GraedelThoWil02}. 

Strategies shall also be represented as finite-state machines.
A \emph{Moore} machine over an input alphabet $\Sigma$ and an output
alphabet $\Gamma$ is described by a tuple
$(M, m_0, \mu, \nu)$ consisting of 
a finite set $M$ of \emph{memory states} with an initial state $m_0$, 
a memory \emph{update} function $\mu: M \times \Sigma \to M$ 
and an \emph{output} function $\nu: M \to \Gamma$ defined on memory
states. 
Intuitively, the machine starts in the initial memory state~$m_0$,   
and proceeds as follows: in state $m$, upon reading an input symbol~$x
\in \Sigma$, 
it updates its memory state to  $m' := \mu( m, x )$ and then
outputs the letter $\nu( m )$. 
Formally, the update function $\mu$ is extended to input words in $\Sigma^*$
by setting, $\mu( \varepsilon ) := m_0$, for the empty word, and
by setting, $
  \mu(x_0 \dots x_{\ell-1} x_{\ell}) := \mu(\,
  \mu( x_0  \dots x_{\ell-1}), \, x_{\ell}),$ 
for all nontrivial words $x_0 \dots x_{\ell-1} x_{\ell}$.
This gives rise to the function $M: \Sigma^* \to \Gamma^*$
\emph{implemented} by $M$, defined by
$M(x_0, \dots, x_{\ell}) := 
\nu(\, \mu(\, x_0 \dots x_{\ell} ) \,).$
A \emph{strategy automaton} for player~$i$ on a game graph~$G$, is a
Moore machine~$M$ with input alphabet $B^i$ and output alphabet $A^i$. 
The strategy implemented by $M$ is defined as
$s^i(v_0, \dots, v_{\ell}) := M( \beta^i( v_0 \dots v_{\ell} ))$, for all~$\ell > 0$.
A~\emph{finite-state strategy} is one that can be implemented by a
strategy automaton. 

Sometimes it is convenient to refer to 
Mealy machines rather than Moore machines. These are finite-state
machines of similar format, with the only difference that the output 
function $\nu: M \times \Sigma \to \Gamma$ is defined on transitions 
rather than their target state.

In the following we will refer to several classes $\mathcal{C}$ of
finite games, always assuming that winning conditions are given as
$\omega$-regular languages.
The \emph{finite-state synthesis problem} for a class $\mathcal{C}$
 is the following: Given a game $\calG \in \mathcal{C}$,
\begin{enumerate}[(i)]
\item decide whether $\calG$ admits
  a finite-state distributed winning strategy, and
\item if yes, construct a profile of finite-state machines that
  implements a distributed winning strategy for $\calG$.
\end{enumerate}
We refer to the set of distributed (finite-state) winning strategies 
for a given 
game $\calG$ as the (finite-state) \emph{solutions} of $\calG$. 
We say that 
the synthesis problem is \emph{finite-state solvable} 
for a class~$\mathcal{C}$ if every game $\calG \in \mathcal{C}$ 
that admits a solution also admits a finite-state solution, and 
if the above two synthesis tasks can be accomplished for all 
instances in~$\mathcal{C}$.

\section{Static Information Hierarchies}

\subsection{Hierarchical observation}

We set out from the basic pattern of hierarchical information
underlying the decidability results cited in the
introduction.
These results rely on a positional interpretation of information,
i.e., on observations.

\begin{definition}
A game graph yields \emph{hierarchical observation} if there exists a total
order~$\preceq$ among the players such that whenever $i \preceq j$, then for
all pairs $v, v'$ of positions,
$\beta^i( v ) = \beta^i( v' )$ implies $\beta^j( v ) = \beta^j ( v' )$
\end{definition}

In other words, if~$i \preceq j$, then the observation of player~$i$ determines the
observation of player~$j$. An example of such a situation is illustrated in 
Figure~\ref{sfig:hierarchical-observation}.

\cite{PetersonRei79} study a 
game with players organised in a hierarchy, such that each player~$i$ sees the data 
observed by player~$i-1$. The setting is actually generic for games with reachability winning conditions
and the authors show that 
winning strategies can be synthesised in $n$-fold exponential time 
and this complexity is unavoidable.
Later, \cite{PnueliRos90} consider a similar model in the context of distributed systems with 
linear-time specifications given by finite automata on infinite words. 
Here, the hierarchical organisation is represented by a pipeline architecture which allows each process to 
send signals only to the following one. The authors show that the distributed synthesis problem for such a 
system, is solvable via an automata-theoretic technique. This technique is 
further extended by \cite{KupfermanVar01} to more general, branching-time specifications.
The key operation of the construction is that of \emph{widening} -- a finite-state
interpretation of strategies for a
less-informed player~$j$ within the strategies of
a more-informed player~$i \preceq j$. 
This allows to first solve a game as if all the
moves were performed by the most-informed player, which comes first
in the order $\preceq$, and successively
discard solutions that cannot be implemented by the less-informed
players, i.e., those which involve strategies that 
are not in the image of the widening interpretation. 

The automata-theoretic method  
for solving the synthesis problem on pipeline architectures, 
due to \cite{PnueliRos90} and \cite{KupfermanVar01},
can be adapted directly 
to solve the synthesis problem for games with hierarchical observation.
Alternatively, the solvability result follows from  
the reduction of games with hierarchical observations to pipeline architectures presented
in Theorem~\ref{thm:hierarchical-obs-pipeline}, as a part of our discussion on 
games and architectures. 

\begin{theorem}[\cite{PnueliRos90,KupfermanVar01}]\label{thm:KV01}
For games with hierarchical observation, the synthesis problem is
finite-state solvable.
\end{theorem}

\begin{figure}[tp]
\begin{center}
\subfigure[hierarchical observation]{ 
    \label{sfig:hierarchical-observation}
\begin{tikzpicture}[xscale=.8]

\node[state3P, initial above] 		(0) at (3,4)         {$\gow$ \nodepart{two} $v_0$ \nodepart{three} $\gow$};

 \node[state3P]              (00) at (1.5,2.5) {$ \gow$ \nodepart{two} $v_1$ \nodepart{three} $\gow$};
 \node[state3P]              (01) at (4.5,2.5) {$\gob $ \nodepart{two} $v_2$ \nodepart{three} $\gow$};

  \node[state3P]              (000) at (1.5,1) {$\gow$ \nodepart{two} $v_3$ \nodepart{three} $\gow$};

 \node[state3P]              (010) at (3.5,1) {$\gob$ \nodepart{two} $v_4$ \nodepart{three} $\gob$};
 \node[state3P]              (011) at (5.5,1) {$ \goc $ \nodepart{two} $v_5$ \nodepart{three} $ \gob$};

 \path (0) edge (00) edge (01);
 \path (00) edge (000);
 \path (01) edge (010) edge (011);
 \draw[-,dotted] (000) -- (1.5,0.3);
 \draw[-,dotted] (010) -- (3.5,0.3);
 \draw[-,dotted] (011) -- (5.5,0.3);
 
\end{tikzpicture}
}
 \hspace*{.3cm}
 \subfigure[static hierarchical information]{ 
    \label{sfig:hierarchical-info}
\begin{tikzpicture}[xscale=.8]

  \node[state3P, initial above] 				(0) at (3,4)         {$\gow $ \nodepart{two} $v_0$ \nodepart{three} $\gow$};
  
  \node[state3P]              (00) at (1.5,2.5) {$ \gow$ \nodepart{two} $v_1$ \nodepart{three} $\gow$};
  \node[state3P]              (01) at (4.5,2.5) {$\gob $ \nodepart{two} $v_2$ \nodepart{three} $\gow$};

  \node[state3P]              (000) at (1.5,1) {$\gow$ \nodepart{two} $v_3$ \nodepart{three} $\gow$};

  \node[state3P]              (010) at (3.5,1) {$\gow$ \nodepart{two} $v_4$ \nodepart{three} $ \gob$};
  \node[state3P]              (011) at (5.5,1) {$ \gob $ \nodepart{two} $v_5$ \nodepart{three} $ \gob$};

  \path (0) edge (00) edge (01);
  \path (00) edge (000);
  \path (01) edge (010) edge (011);
  \draw[-, dotted] (000) -- (1.5,0.2);
  \draw[-, dotted] (010) -- (3.5,0.2);
  \draw[-, dotted] (011) -- (5.5,0.2);
 
\end{tikzpicture}
}
\end{center}
\caption{Basic patterns of hierarchical information: 
game positions show the observation of player~$1$ (left) 
and player~$2$ (right); 
the name of the position (middle) is unobservable}
\label{fig:hierarchieS}
\end{figure}
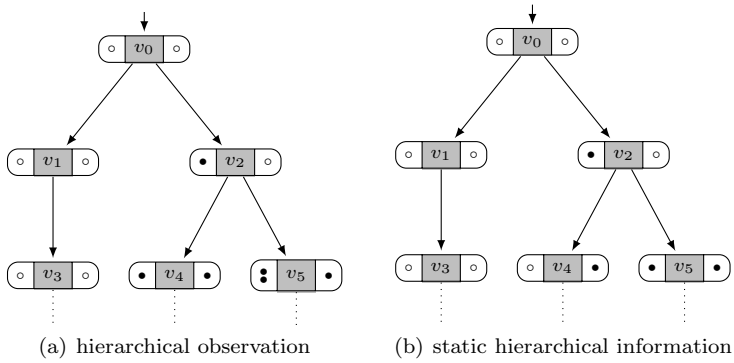

\subsection{Incorporating perfect recall}

In a first step, we extend the notion of hierarchical 
information to incorporate
the power of perfect recall that players have. 
While maintaining the requirement of a fixed order, we now ask that
the \emph{information set} of a player determines the information sets of
those who follow in the order, as in the example of Figure~\ref{sfig:hierarchical-info}.

\begin{definition}
A game graph yields \emph{(static) hierarchical information} 
if there exists a total order~$\preceq$ among the players such 
that, for all histories~$\pi$, if $i \preceq j$, then 
$P^i( \pi ) \subseteq P^j( \pi )$.
\end{definition}

The following lemma provides an operational characterisation of the
condition. 
We detail the proof, as its elements will be used later.

\begin{lemma}\label{lem:translate-observation}
A game graph $G$ yields static hierarchical information if, and only if, 
for every pair $i \preceq j$ of players, there exists a Moore machine that outputs
$\beta^j( \pi )$ on input $\beta^i(\pi)$, for every history $\pi$ in $G$.
\end{lemma}

\begin{proof}
For an arbitrary game graph~$G$, let us denote the relation 
between the
observations of two players~$i$ and $j$ along the histories in~$G$ 
by \begin{align*}T^{ij} := \{ (\beta^i( \pi ), \beta^j( \pi )) \in (B^i \times
B^j)^* ~|~\pi \in \Hist( G )
\}.
\end{align*}

Let us first assume that  there exists a Moore machine that recognises $T^{ij}$. 
Then $T^{ij}$ is actually a function and hence 
$\pi \sim^i \pi'$ implies 
\begin{align*}
  \beta^j( \pi ) = T^{ij}( \beta^i ( \pi ) ) = T^{ij}( \beta^i( \pi' )) = \beta^j( \pi' ),
\end{align*}  
and we can conclude that $\pi \sim^j \pi'$. 

To see that the converse implication holds, notice that~$T^{ij}$ is a regular relation, 
recognised by the game graph $G$ viewed as a finite-word
automaton $A_G^{ij}$ over
the alphabet of observation pairs $B^i \times B^j$. 
Concretely, consider the nondeterministic automaton~$A_G^{ij} := (V, B^i\times B^j,v_0, \Delta, V)$ 
on states corresponding to positions
of $G$, with initial state $v_0$ and 
with transitions  $( v, (b^i,b^j), v' ) \in \Delta$ if 
there
exists a move $(v, a, v') \in E$ such that
$\beta^i( v' ) = b^i$ and $\beta^j( v' ) = b^j$; all states are accepting.
Further, let~$M^{ij}$ be the automaton obtained by determinising $A_G^{ij}$ and
trimming the result, that is, removing all states that do not
lead to an accepting state. 

Now assume that the game graph $G$ at the outset yields static hierarchical information. Then, the relation~$T^{ij}$
recognised by $M^{ij}$ is functional, and hence $M^{ij}$ is deterministic in the
input component $i$: for any state $v$ there exists precisely one
outgoing transition along each observation $b^i \in B^i$. In other words, $M^{ij}$ is
a Mealy machine, which we can transform
into an equivalent Moore machine, as desired.
\qed
\end{proof}

\begin{theorem}\label{thm:static-decidable} 
For games with static hierarchical information, the synthesis problem 
is finite-state solvable.
\end{theorem}
\begin{proof}
Intuitively, we transform an arbitrary game graph $G = (V, E, \beta)$ with static hierarchical
\emph{information} into one with hierarchical \emph{observation}, 
by taking the synchronised product of~$G$ with automata 
that signal to each player~$i$ the
observations of all players~$j \succeq i$. 
We shall see that this
preserves the solutions to the distributed synthesis problem, 
for any winning condition on $G$. 

To make the construction precise, let us fix a pair $i \preceq j$ of players, and consider
the Moore machine $M^{ij} = (M, m_0, \mu, \nu)$ as in the proof of
Lemma~\ref{lem:translate-observation}, which translates
the observations $\beta^i( \pi )$ into 
$\beta^j( \pi )$, for every history $\pi$ in $G$.  
We define the product $G \times M^{ij}$ as a new game graph with the
same sets of actions as~$G$, and the same observation alphabets $(B^k)_{k \neq i}$,
except for player~$i$, for which we expand the alphabet to 
$B^i \times B^j$ to also include observations of player~$j$. 
The new game is over positions in 
$V \times M$ with moves $( (v, m), a, (v', m') ) $ if $(v, a, v') \in E$ and 
$\mu( m, \beta^i( v ) ) = m' $.
The observations for
player~$i$ are given by $\beta^i( v, m ) = ( \beta^i( v ), \nu ( m ))$, 
whereas they remain unchanged for all other 
players~$\beta^k( v, m ) = \beta^k( v )$, for all $k \neq i$.

The obtained product graph is equivalent to the original game graph $G$, in the
sense that they have the same tree unravelling, and
the additional components in the observations of player~$i$
(representing observations of player~$j$, given by the Moore
machine~$M^{ij}$) are already determined by his own observation
history, so player~$i$ 
cannot distinguish any pair of histories in the new game that he could
not distinguish in the original game. 
Accordingly, the strategies on the expanded game graph $G \times
M^{ij}$ correspond to
strategies on~$G$, such that the 
outcomes of any distributed strategy are preserved.
In particular, for any winning condition over $G$, a distributed strategy 
is winning in the original game if, and only if, it is winning in the
expanded game $G \times M^{ij}$.
On the other hand, the (positional) observations of Player $i$ in the
expanded game determine the observations of Player $j$.

By applying the transformation for each pair $i \preceq j$ of players
successively, we obtain a game graph that yields hierarchical
observation. Moreover, every winning condition on~$G$ induces a winning condition 
on (the first component of positions in)~$G \times M^{ij}$ 
such that the two resulting games have the same 
winning strategies.
Due to Theorem~\ref{thm:KV01}, we can thus conclude that,
under $\omega$-regular winning conditions, the synthesis problem is
finite-state solvable for games with static
hierarchical information.
\qed
\end{proof}

To decide whether a given game graph yields static hierarchical
information, the collection of Moore machines constructed according to 
Lemma~\ref{lem:translate-observation}, for all players~$i,j$, may be
used as a witness. However, this yields an inefficient procedure,
as the determinisation of a functional transducer 
involves an exponential blowup; precise bounds for such translations
are given by~\cite{WeberKlemm95}.
More directly, one could
verify that each of the transductions $A_G^{ij}$ relating observation
histories of Players~$i,j$, as defined in the proof of 
Lemma~\ref{lem:translate-observation},
is functional. This can be done in
polynomial time using, e.g., the procedure described in~\cite{BealEtAl00}. 

We can give a precise bound in terms of nondeterministic complexity.

\begin{lemma}\label{lem:deciding-hierarchical-static}
The problem of deciding whether a game yields static hierarchical
information is \textsc{NLogSpace}-complete.
\end{lemma}

\begin{proof}
The complement problem\,---\,of verifying
that for a given game there exists a pair of players $i,j$ 
that cannot be ordered in either way\,---\,%
is solved by the following nondeterministic procedure: 
Guess a pair~$i,j$ of players, then check that $i \not\preceq j$, by
following nondeterministically a pair of histories~$\pi \sim^i \pi'$, 
such that $\pi \not\sim^j \pi'$; symmetrically, check that $j
\not\preceq i$. 
The procedure requires only logarithmic space for maintaining
pointers to four positions while keeping track of the
histories. 
Accordingly, the complement problem is in $\NLOGSPACE$, and since the
complexity class is closed under complementation (\cite{Immerman88,Szelepcsenyi88}), 
our decision problem of 
whether a game yields static hierarchical
information also belongs to $\NLOGSPACE$.

For hardness, we reduce the
emptiness problem for nondeterministic finite automata, known to be $\NLOGSPACE$-hard 
(\cite{Jones75}),  
to the problem of verifying that the following game for two players
playing against Nature on the graph of the automaton
yields hierarchical information:
Nature chooses a run in the automaton, the players can only observe the input
letters, unless an accepting state is reached; 
if this happens, Nature sends to each
player privately one bit, which violates the condition of
hierarchical information. Thus, the game has hierarchical
information if, and only if, no input word is accepted. 
\qed 
\end{proof}

\subsection{Signals and game transformations}

Functions that return information about the current history, 
such as those constructed in the proof of
Lemma~\ref{lem:translate-observation} will be a useful tool in our 
exposition, especially when the information can
be made observable to certain players without changing the game.

Given a game graph~$G$, a~\emph{signal} is a function defined on the
set of histories in~$G$, or on the set of observation histories of
some player~$i$.
We say that a signal~$f: \Hist( G ) \to \Sigma$ is
information-consistent for player~$i$ if any two histories that are
indistinguishable to~player~$i$ have the same image under~$f$; 
in particular, strategies are information-consistent signals.
A finite-state signal is one  
implemented by a Moore machine. Any 
finite-state signal $f: \Hist( G ) \to \Sigma$ can also be 
implemented by a Moore machine $M^i$ over the observation alphabet $B^i$, 
such that that $M ( \pi ) = M^i( \beta^i ( \pi ))$ for every history~$\pi$.
The \emph{synchronisation} of $G$ with a finite-state signal $f$ is
the expanded game graph $(G, f)$ obtained by taking the synchronised
product $G \times M$, as described in the proof of
Lemma~\ref{lem:translate-observation}. 
In case~$f$ is
information-consistent for player~$i$, it can be made
\emph{positionally observable} to this player,
without changing the game essentially. 
Towards this, we consider the game graph
$(G, f^i)$ that expands $(G, f)$ with an additional observation
component~$f^i( v )$ for player~$i$ at every position~$v$, 
such that $f( \pi ) = f^i( v )$ 
for each history $\pi$ that ends at $v$.   
The game graph $(G, f^i)$ is 
\emph{finite-state equivalent} to~$G$, in the sense 
that every strategy for $G$ maps via finite-state transformations to a strategy for $(G, f^i)$ 
with the same outcome and vice versa. 
Indeed, any strategy for $G$ is readily
a strategy with the same outcome for~$(G, f^i)$  and, conversely, 
every strategy profile~$s$ in~$(G, f^i)$ can be 
synchronised with the Moore machines implementing the signals~$f^i$
 for each player~$i$, to yield a finite-state 
strategy profile~$s'$ for $G$ with the same outcome as~$s$.
In particular, the transformation preserves solutions to the
finite-state synthesis problem under any winning condition.

\section{Dynamic Hierarchies}

In this section,
we maintain the requirement on the information sets
of players to be totally ordered at every
history. However, in
contrast to the case of static hierarchical information, we allow 
the order to depend on the history and to change dynamically along a play.  
Figure~\ref{sfig:dynamic-hierarchical} shows an example of such a situation: 
at the history reaching $v_2$, player~$1$ is more informed than player~$2$, however, the order switches 
when the play proceeds to position~$v_4$, for instance.

\begin{definition} 
A history~$\pi$ in a game yields \emph{hierarchical information} if 
the information sets $\{ P^i( \pi)~|~ i \in N \}$ are 
totally ordered by inclusion.
A game graph yields \emph{dynamic} hierarchical information
if every history yields hierarchical information.
\end{definition}

We first observe that, for every finite game, 
the set of histories that yield hierarchical information is
regular. We detail here the construction of an automaton for the 
complement language, which will also be of later use.

\begin{lemma}\label{lem:automaton-hierarchical-history}
For every finite game graph~$G$, we can construct a nondeterministic finite
automaton that accepts the histories in~$G$ that do not yield
hierarchical information. 
If~$G$ has $n$ players and $|V|$ positions, 
the number of automaton states is at most $2n^2|V|^2$.
\end{lemma}

\begin{proof} 
Let us fix a game graph~$G$. 
A history~$\pi$ in~$G$ fails to yield hierarchical information if there
are two players with incomparable information sets at~$\pi$. 
To verify this, we construct an automaton
that chooses nondeterministically 
a pair~$i,j$ of players, then, while reading the input~$\pi$, 
it guesses a pair $\pi'$,~$\pi''$ of histories  
such that $\pi' \sim^i\pi$ and $\pi'' \sim^j \pi$ and updates two flags
indicating whether $\pi' \not\sim^i \pi''$ or $\pi' \not\sim^j \pi''$; 
the input is accepted if both flags are set. Hence, a 
word $\pi \in V^*$ that corresponds to a history in~$G$ is accepted if, and only if, 
the corresponding history does not yield hierarchical information.
\footnote{Notice that the automaton may also accept words that do not 
correspond to game histories; 
to avoid this, we can take the synchronised product with the game graph~$G$ 
and obtain an automaton that recognises precisely the set of histories 
that do not yield hierarchical information.}

In its states, the constructed automaton stores the indices of the two players~$i,j$, 
a pair of game positions to keep track of the witnessing histories $\pi'$ and $\pi''$, 
and a two-bit flag to record whether the current input prefix is distinguishable 
from $\pi'$ for player~$j$ or from~$\pi''$ for player~$i$. 
Clearly, it is sufficient to consider each pair of players only once, hence,  
the automaton needs at most $4\frac{n(n-1)}{2}|V|^2 $ states, that is, 
less than $2n^2|V|^2$.
\qed
\end{proof}

To decide whether a given game graph~$G$ 
yields dynamic hierarchical information, 
we may check whether the automaton described in 
Lemma~\ref{lem:automaton-hierarchical-history} accepts all histories 
in~$G$. However, more efficient than constructing this automaton, 
we can use a nondeterministic procedure 
similar to the one of 
Lemma~\ref{lem:deciding-hierarchical-static} 
to verify on-the-fly if there exists a history at which 
the information sets of two players are incomparable: 
guess two players~$i,j$
and three histories $\pi\sim^i\pi'$ and $\pi'' \sim^j \pi$, such that
$\pi' \not\sim^i \pi''$ and $\pi' \not\sim^j \pi''$. 
Obviously, the lower bound 
from~Lemma~\ref{lem:automaton-hierarchical-history} is preserved.

\begin{lemma}\label{lem:deciding-hierarchical-dynamic}
The problem of deciding whether a game graph yields dynamic hierarchical
information is \textsc{NLogSpace}-complete.
\end{lemma}

\begin{figure}
  \begin{center}
  \subfigure[dynamic hierarchical information]{ 
    \label{sfig:dynamic-hierarchical}
    \begin{tikzpicture}[xscale=0.8]
      
      \node[state3P, initial above] 				(0) at (3,4)         {$\gow $ \nodepart{two} $v_0$ \nodepart{three} $\gow$};
      
      \node[state3P]              (00) at (1.5,2.5) {$ \gow$ \nodepart{two} $v_1$ \nodepart{three} $\gow$};
      \node[state3P]              (01) at (4.5,2.5) {$\gob $ \nodepart{two} $v_2$ \nodepart{three} $\gow$};
            
      \node[state3P]              (000) at (1.5,1) {$\gow$ \nodepart{two} $v_3$ \nodepart{three} $\gow$};
      
      \node[state3P]              (010) at (3.5,1) {$\gow$ \nodepart{two} $v_4$ \nodepart{three} $\gob$};
      \node[state3P]              (011) at (5.5,1) {$\gow$ \nodepart{two} $v_5$ \nodepart{three} $ \goc$};

      \path (0) edge (00) edge (01);
      \path (00) edge (000);
      \path (01) edge (010) edge (011);
      \draw[-, dotted] (000) -- (1.5,.3);
      \draw[-, dotted] (010) -- (3.5,.3);
      \draw[-, dotted] (011) -- (5.5,.3); 
    \end{tikzpicture}
    \hspace*{.3cm}
  }
  \subfigure[recurring hierarchical information]{ 
    \label{sfig:recurring-hierarchical}
      \begin{tikzpicture}[xscale=0.8]

        \node[state3P, initial above] 				(0) at (3,4)         {$\gow $ \nodepart{two} $v_0$ \nodepart{three} $\gow$};
        
        \node[state3P]              (00) at (1.2,2.5) {$\gow$ \nodepart{two} $v_1$ \nodepart{three} $\gob$};
        \node[state3P]              (01) at (3,2.5) {$\gow$ \nodepart{two} $v_2$ \nodepart{three} $\gow$};
        \node[state3P]              (02) at (5,2.5) {$\gob $ \nodepart{two} $v_3$ \nodepart{three} $\gow$};

        \node[state3P]              (000) at (2.1,1) {$\gow$ \nodepart{two} $v_4$ \nodepart{three} $\gow$};

        \node[state3P]              (010) at (5,1) {$\gob$ \nodepart{two} $v_5$ \nodepart{three} $\gob$};

        \path (0) edge (00) edge (01) edge (02);
        \path (00) edge (000);
        \path (01) edge (000);
        \path (02) edge (010);

        \path (000) edge [bend left=100, looseness=1.8] (0);
        \path (010) edge [bend right=90, looseness=1.7] (0);
        \node[draw=none] (x) at (1.5,.3){};
        
   \end{tikzpicture}
 }
 
\end{center}
\caption{More patterns of hierarchical information}
\label{fig:hierarchieS}
\end{figure}
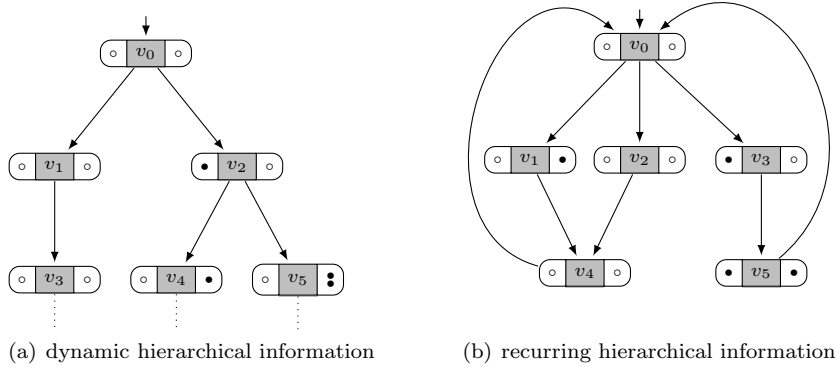

In the remainder of the section, we show that, 
 under the more liberal condition of dynamic hierarchical information, 
distributed games are still
decidable. 

\begin{theorem}\label{thm:dynamic-decidable} 
For games with dynamic hierarchical information, the synthesis problem 
is finite-state solvable.  
\end{theorem}

For the proof, we transform an arbitrary game~$\calG$ with 
dynamic hierarchical information into one with static
hierarchical information, among a different set of $n$
\emph{shadow} players $1', \dots, n'$, where
each shadow player~$i'$ 
plays the role of the $i$-most informed player in the original
game, in a sense that we will make precise soon. 
The information sets of the shadow players follow their nominal
order, that is, if $i < j$ then $P^{i'}( \pi ) \subseteq P^{j'}( \pi
)$. 
The resulting shadow game inherits the graph structure of the orginal game,
and we will ensure that, for every history~$\pi$,
\begin{enumerate}[(i)]

\item each shadow player~$i'$ has the same
information (set) as the $i$-most informed actual player, and
\item each shadow player~$i'$ has the same
choice of actions as the $i$-most informed actual player.
\end{enumerate}
This shall guarantee that the shadow game preserves 
the winning status of the original game. 

The construction proceeds in two phases. 
Firstly, we expand the game graph~$G$ so that 
the correspondence between actual and shadow players does not depend
on the history, but only on the current position. This is done
by synchronising~$G$ with a finite-state machine that
signals to each player his rank in the information hierarchy at the
current history. 
Secondly, we
modify the game graph, where the shadow-player correspondence is recorded as a
positional attribute,
such that the observation of each player 
is received by his shadow player, at every position; similarly,   
the actions of each player are transferred to his shadow player.
Finally, we show how finite-state 
winning strategies for the shadow game can be re-distributed 
to the actual players to yield a winning profile of finite-state
winning strategies for the original game.

\subsection{Information rank signals}

For the following, let us fix a game $\calG$ with dynamic hierarchical information with
the usual notation. For a history~$\pi$, we write $\preceq_\pi$ for
the total order among players induced by the inclusions between their information sets
at~$\pi$.
To formalise the notion of an~$i$-most informed player, 
we use the shortcut $i \approx_\pi j$ to denote that $i \preceq_\pi j$ and 
$j \preceq_\pi i$; 
likewise, we write $i \prec_\pi j$ to denote that $i \preceq_\pi j$ and not $j
\preceq_\pi i$. 

Then, the \emph{information rank} of player~$i$ on the game graph $G$ is a signal $\rank^i: \Hist(
G ) \to N$ defined by
\begin{align*}
\textrm{rank}^i( \pi ) := |\{ j \in N ~|~ 
  j \prec_\pi i \text{ or } 
  ( j < i \text{ and } j \approx_\pi i )\, \}|.
\end{align*}
Likewise, we define the \emph{order} of player~$i$ \emph{relative} to
player~$j$ as a Boolean signal
 $\preceq^i_j: \Hist( G ) \to \{0, 1\}$ 
with $\preceq^i_j( \pi ) = 1$ if, and only if, $i \preceq_\pi j$.

\begin{lemma}\label{lem:rank-observable}
The  information rank of each player~$i$ and his order relative to
any player~$j$  are finite-state signals 
that are information-consistent to player~$i$. 
\end{lemma}
\begin{proof} We detail the argument for the rank, 
  the case of relative order is similar and simpler.

Given a game $\calG$ as in the statement,
let us verify that the signal $\rank^i $ is information-consistent,
for each player~$i$. Towards this, consider two histories
$\pi \sim^i \pi'$ in $G$, 
and suppose that some player~$j$ does not count for the rank of $i$ at $\pi$, in the
sense that either $i \prec_\pi j$ or ($i \approx_\pi j$ and $i <
j$) --- in both cases, it follows that $\pi \sim^j \pi'$, hence 
$P^j( \pi ) = P^j( \pi' )$, which implies that $j$ does not count for
the rank of $i$ at $\pi'$ either. Hence, the set of players that
count for the rank of player~$i$ is the same at $\pi$ and at
$\pi'$, which means that $\rank^i( \pi ) = \rank^i( \pi' )$.

To see that the signal $\rank^i$ can be implemented by a
finite-state machine, we first build, for every pair $i, j$ of players,
a nondeterministic automaton $A_i^j$ that accepts the histories $\pi$ 
where $j \prec_\pi i$, by guessing a history $\pi' \sim^i \pi$ 
and verifying that $\pi' \not\sim^j \pi$. To accept the histories
that satisfy $i \approx_\pi j$, we take the product of the automata
$A_i^j$ and $A_j^i$ for $i \preceq_\pi j$ and $j \preceq_\pi i$ and accept if both
accept. Combining the two constructions allows us to describe, 
for every player~$j$, an automaton $A_j$ 
to recognise the set of histories at which~$j$ counts for $\rank^i( \pi )$.

Next, we determinise each of the automata $A_j$ and take appropriate
Boolean combinations to obtain a Moore machine $M^{i}$ with input alphabet $V$ and 
output alphabet $\pset{N}$, which upon reading a
history $\pi$ in $G$, outputs the set of players that count for 
$\rank^i( \pi )$. Finally we replace each set in the output of $M^{i}$
by its size to obtain a Moore machine
that returns on input $\pi \in V^*$, the rank of 
player~$i$ at the actual history~$\pi$ in $G$. 

As we showed that $\rank^i$ is an information-consistent
signal, we can conclude that there exists a Moore machine that inputs
observation histories $\beta^i( \pi )$ of player~$i$ and outputs
$\rank^i( \pi )$.
\qed
\end{proof}

One consequence of this construction is that we can view 
the signals~$\rank^i$ and~$\preceq^i_j$ as attributes of positions rather than
properties of histories. 
Accordingly, we can assume 
without loss of generality that the observations of each player~$i$ 
have an extra $\rank$ component taking values in~$N$ and that the symbol 
$j$ is observed at history~$\pi$ in this component if, and
only if, $\rank^i(\pi) = j$. When referring to the positional attribute
$\preceq^i_j$ at $v$, it is more convenient to write 
$i\preceq_v j$ rather than $\preceq_j^i$.
 
\subsection{No crossing}

As we suggested in the proof outline,
each player~$i$ and his shadow player, identified by the observable
signal~$\rank^i$, should be equally informed. 
To achieve this, we will let the observation of 
player~$i$ be received by his shadow, in every round of a play. 
However, since the rank of players, and hence the
identity of the shadow, changes with the history, 
an information loss can occur 
when the information order between two players, say~$1 \prec 2$ along a
move is swapped to become~$2 \prec 1$ in the next round. 
Intuitively, the observation received by player~$2$ after this move 
contains one piece of information that allows him to
catch up with~player~$1$, and another piece of information to
overtake~player~$1$. Due to their rank change along the move, the
players would now also change shadows. Consequently, 
the shadow of $1$ at the target position, who was
previously as (little) informed as player~$2$, just receives the
new observation of player~$1$, but he may miss the piece of
information that allowed player~$2$ to catch up (and which player~$1$
had). Figure~\ref{sfig:crossing} pictures such a situation.

\begin{figure}
\begin{center}
\subfigure[crossing: $1 \prec 2$ at $v_2$ to $2 \prec 1$ at $v_4$]{
  \label{sfig:crossing}
  \begin{tikzpicture}[xscale=.7]

\node[state3P, initial above] 				(0) at (6,8)         {$\circ$ \nodepart{two} $v_0$ \nodepart{three} $\circ$};

 \node[state3P]              (00) at (4,6) {$\circ$ \nodepart{two} $v_1$ \nodepart{three} $\circ$};
 \node[state3P]              (01) at (8,6) {$\bullet$ \nodepart{two} $v_2$ \nodepart{three} $\circ$};

  \node[state3P]              (000) at (4,4) {$\circ$ \nodepart{two} $v_3$ \nodepart{three} $\circ$};

 \node[state3P]              (010) at (7,4) { $\circ$\nodepart{two} $v_4$ \nodepart{three} $\bullet$};
 \node[state3P]              (011) at (9,4) {$\circ$ \nodepart{two} $v_5$ \nodepart{three} $\goc$}; 
 
 \path (0) edge (00) edge (01);
 \path (00) edge (000);
 \path (01) edge (010) edge (011);
 

\end{tikzpicture}
\hspace*{.3cm}
}
\hspace*{.3em}
\subfigure[half-step lookahead]{
\label{sfig:lookahead}
\begin{tikzpicture}[xscale=0.8]

\node[state3P, initial above] 				(0) at (6,8)         {$\circ$ \nodepart{two} $v_0$ \nodepart{three} $\circ$};

\node[state, rectangle, dashed] 				(0') at (5,7)         {$\gob,\gow|~~~~$};

 \node[state3P]              (00) at (4,6) {$\circ$ \nodepart{two} $v_1$ \nodepart{three} $\circ$};
 \node[state, rectangle, dashed]              (01') at (7,7) {$\gob,\gow|~~~~$}; 
 
 \node[state3P]              (01) at (8,6) {$\bullet$ \nodepart{two} $v_2$ \nodepart{three} $\circ$};
  \node[state, rectangle, dashed]              (000') at (4,5) {$~~~\gow|\gow~~~$};   
       
  \node[state3P]              (000) at (4,4) {$\circ$ \nodepart{two} $v_3$ \nodepart{three} $\circ$};

 \node[state3P]              (010) at (7,4) { $\circ$\nodepart{two} $v_4$ \nodepart{three} $\bullet$};
 \node[state3P]              (011) at (9,4) {$\circ$ \nodepart{two} $v_5$ \nodepart{three} $\goc$}; 

 \node[state, rectangle, dashed]              (010') at (7.3,5) {$~~~~|\gob,\goc$};
 \node[state, rectangle, dashed]              (011') at (8.7,5) {$~~~~|\gob,\goc$};

  \path (0) edge (0') edge (01') ;
 \path (0') edge (00);

 \path (01') edge (01);
 \path (00) edge (000');
  \path (000') edge (000);
  
 \path (01) edge (010') edge (011');
 \path (010') edge (010);
 \path (011') edge (011);

\end{tikzpicture}
}
\end{center}
\caption{Eliminating crossings}
 \label{fig:cross-free}
\end{figure}
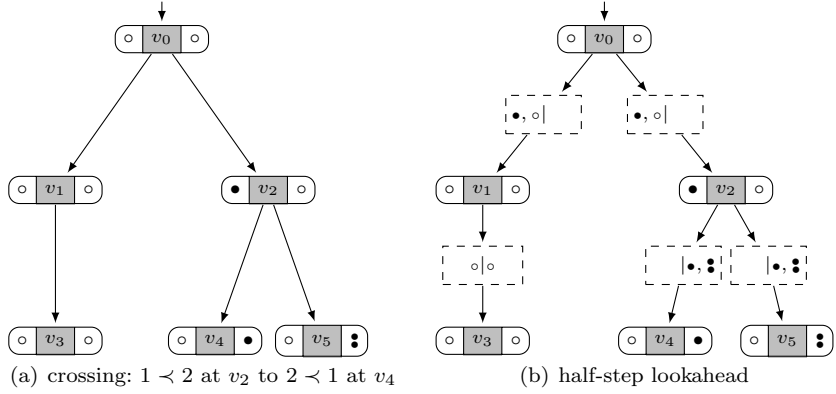

We describe a transformation to eliminate the crossings due to switches in the
information order, such that this artefact does no longer occur.
Formally, for a play $\pi$ in~a game, we say that Player $i$ and $j$ cross at stage $\ell$
if $P^i( \pi_\ell ) \subsetneq P^j( \pi_\ell)$
and $P^j( \pi_{\ell+1} ) \subsetneq P^i( \pi_{\ell+1})$.
We say that a game with dynamic hierarchical information
is \emph{cross-free} if there are no crossing players in any play.

\begin{lemma}\label{lem:nocross}
Every game with dynamic hierarchical information is finite-state
equivalent to a game that is cross-free.
\end{lemma}

\begin{proof}
Let~$G$ be a game graph with dynamic hierarchical information.
We define a signal for each pair of players~$i, j$ that represents
the knowledge that player $j$ has about the current
observation of player~$i$. 
If this signal is made observable to
Player $i$ only at histories~$\pi$ at which $i \preceq_\pi j$,  
the game remains essentially unchanged,
as players only receive information from less-informed players, 
which they could hence deduce from their observation. 
Concretely, we define the signal $\lambda_j^i: V^* \to \pset{B^i}$
by 
\begin{align*}
  \lambda_j^i( \pi ) := \{ \beta^i( v' )~:~ v' \text{ is the last state of some
    history } \pi' \in P^j( \pi ) \}.
\end{align*}
Clearly, this is a finite-state signal. 

Now we look at the 
synchronised product of $G$ with the signals~$(\lambda_j^i)_{i,j \in N}$ and the
relative-order signal $\preceq^i_j$ constructed in the proof of
Lemma~\ref{lem:rank-observable}. In the resulting game graph, the signal value $\lambda_j^i( \pi )$ at a 
history~$\pi$ that ends at a position~$w$ is represented by the 
position attribute $\lambda_j^i( w )$. We add
to every move $(v, a, w)$ an intermediary position~$u$, at which we
assign, for every player~$i$ the observation
$\{\lambda_j^i( w )~:~ i \preceq_w j \}$.
Intuitively, this can be viewed as a half-step lookahed signal that
player~$i$ receives from player~$j$ who may have been more informed at
the source position~$v$ -- thus the signal is not
necessarily information-consistent for player~$i$. 
Nevertheless, the game  remains essentially unchanged after adding the signal, as the players
cannot react to the received observation before reaching the target~$w$,
at which point the information is readily revealed. 
On the other hand, along moves at
which the information order between players switches, the
intermediary position ensures that the players attain equal
information. The construction is illustrated in Figure~\ref{sfig:lookahead}.

For any game~$\calG$ on~$G$, we adjust the winning condition 
to obtain one for the new game graph, by ignoring the added intermediary
positions. As the added positions have only one successor, the players have no relevant choice, 
so any distributed 
winning strategy for the new game corresponds to one for~$\calG$ and vice versa.
In particular, for $\omega$-regular winning conditions, 
the construction yields a game with no
crossings that is finite-state
equivalent to the original game.
\qed
\end{proof}

\subsection{Shadow players}

We are now ready to describe the construction of the shadow game
associated to~a game~$\calG = (V, E, \beta, W)$ with dynamic hierarchical information.
Without loss of generality, we can assume
that every position in $G$ is marked with the attributes $\rank^i(v)$
and $\preceq_j^i$, for
all players~$i$,~$j$ according to
Lemma~\ref{lem:rank-observable} and that the game graph is cross-free,
according to Lemma~\ref{lem:nocross}. 

The shadow game $\calG' = (V \cup\{ \ominus \}, E', \beta', W)$ is
also played by $n$ players
and has the same winning condition as $\calG$.
The action and the observation alphabet of each shadow player 
consists of the union of the action and observation alphabets of all
actual players. The game graph $G'$ has the same positions as~$G$, 
plus one sink $\ominus$ that absorbs all moves along unused action profiles.
The moves of $G'$ are obtained from $G$ by assigning the actions of each player~$i$ to
his shadow player~$j = \rank^i( v )$ as follows: 
for every move $(v, a, v') \in E$, 
there is a move $(v, x, v') \in E'$ 
labelled with the action profile $x$ obtained by a permutation of~$a$
corresponding to the rank order, that is, $a^i =
x^j$ for $j = \rank^i( v )$, for all players~$i$.
Finally, at every position $v \in V$, the observation of any player~$i$
in the original game~$\calG$ is assigned to his shadow player, that is
$\beta'^j( v ) := \beta^i( v )$, for $j = \rank^i( v )$.  

By construction, the shadow game yields static hierarchical
information, according to the nominal order of the players. 
We can verify, by induction on the length of histories, that for every
history~$\pi$, the information set of player~$i$ at~$\pi$ in~$G$
is the same as the one of his shadow player $\rank^i( \pi )$ in~$G'$.

\medskip

Finally, we show that the distributed synthesis problem for $G$
reduces to the one on $G'$, and vice versa.
To see that~$\calG'$ admits a winning strategy if~$\calG$ does, 
let us fix a distributed strategy~$s$ for the actual players
in~$\calG$. We define~a signal~$\sigma^j: \Hist( G' ) \to A$ 
for each player in~$\calG'$, by setting $\sigma^j( \pi ) := s^i ( \pi )$ 
if $j = \rank^i( \pi )$, for each history~$\pi$. 
This signal is information-consistent for player~$j$,
since, at any history~$\pi$, his information set is the same as for
the actual player~$i$ with $\rank^i( \pi ) = j$, 
and because the strategy of the actual player~$i$ is
information-consistent for himself. Hence, $\sigma^j$ is a strategy for
player~$j$ in~$G'$. 
Furthermore, at every history, the action taken by the shadow player 
$j = \rank^i( \pi )$ has the same outcome as if it was taken
by the actual player~$i$~in~$G$. 
Hence, the set of play outcomes of the profiles $s$ and $\sigma$ are the
same and we can conclude that, if there exists a distributed winning strategy for $G$,
then there also exists one for $G'$. 
Notice that this implication holds under any winning condition, 
without assuming $\omega$-regularity.

For the converse implication, let us suppose that the
shadow game~$\calG'$ admits a winning profile $\sigma$ of finite-state
strategies. We consider, for each actual player~$i$ of $G$, the signal
$s^i: \Hist( G ) \to A^i$ 
that maps every history~$\pi$ 
to the action $s^i( \pi ) := \sigma^j( \pi )$ of the shadow player $j = \rank^i( \pi )$. 
This is a finite-state signal, as we can implement it by synchronising
$G$ with $\rank^i$, the observations of the shadow players,
and the winning strategies $\sigma^j$, for all shadow players~$j$.
Moreover, $s^i$ is
information-consistent to the actual player~$i$, because all histories  
$\pi \in P^i( \pi )$, have the same value $\rank^i( \pi )
=: j$, and, since $\sigma^j$ is 
information-consistent for player~$j$, the actions prescribed by 
$\sigma^j( \pi )$ must
be the same, for all $\pi \in P^j ( \pi) = P^i( \pi )$. 
In conclusion, the signal $s^i$ represents a
finite-state strategy for player~$i$. The 
profile $s$ has the same set of play outcomes as $\sigma$, so
$s$ is indeed a distributed finite-state strategy, as desired. 

In summary, we have shown that any game $\calG$ with dynamic hierarchical
information admits a winning strategy if, and only if, the associated
shadow game with static hierarchical observation admits a finite-state
winning strategy. The latter question is decidable according to
Theorem~\ref{thm:static-decidable}. We showed that for every positive
instance~$G'$, we can construct a finite-state distributed
strategy for $G$. This concludes the proof of Theorem~\ref{thm:dynamic-decidable}.

\section{Transient Perturbations}

As a third pattern of hierarchical information, we consider
the case where incomparable information sets may occur at some
histories along a play, but it is guaranteed that a total order will be
re-established in a finite number of rounds.

\begin{definition}
  A play yields \emph{recurring hierarchical information} if 
 it has infinitely many prefix histories that yield hierarchical
 information.
 A game yields \emph{recurring hierarchical information} if all its
 plays do so. 
\end{definition}

Since the set of histories that yield hierarchical information is 
regular in any finite game, according to Lemma~\ref{lem:automaton-hierarchical-history}, 
it follows that the set of plays that yield recurring hierarchical information is 
$\omega$-regular as well.  

\begin{lemma}\label{lem:automaton-hierarchical-recurring}
  For every finite game, we can construct a deterministic B{\"u}chi automaton
  that recognises the set of plays that yield
  recurring hierarchical information.
  If~$G$ has~$n$ players and $|V|$ positions, 
the number of automaton states is bounded by $2^{O(n^2 |V|^2)}$.
\end{lemma}

\begin{proof} Given a game graph~$G$, 
we can construct a nondeterministic finite 
automaton~$A$ that recognises the set of histories in~$G$ that do 
\emph{not} yield hierarchical information, 
as in Lemma~\ref{lem:automaton-hierarchical-history}.
By determinising the automaton~$A$ via the standard powerset construction, 
and then complementing the set of accepting states,  
we obtain a deterministic 
automaton~$A^\complement$ with at most $2^{2n^2|V|^2}$ states that accepts a word~$\pi \in V^*$ 
if either~$\pi \not \in \Hist( G )$, or~$\pi$ 
represents a history in~$G$ that yields hierarchical information.
Finally, we take the synchronised product~$B$ of $A^\complement$ with the graph~$G$. 
If we now view the resulting automaton as a B{\"u}chi automaton, 
which accepts an infinite words if infinitely many prefixes are accepted 
by~$B$, we obtain a deterministic automaton
that recognises the set of plays in~$G$ that yield
hierarchical information. The number of states in~$B$ is at most 
$|V|\,2^{2 n^2|V|^2}$, hence bounded by $2^{O(n^2 |V|^2)}$.
\qed
\end{proof}

The automaton construction provides an important insight about the 
number of consecutive rounds in which players may have incomparable information. 
Given a play~$\pi$ on a game graph~$G$,
we call a \emph{gap} any interval $[t, t + \ell]$ of
rounds such that the histories of~$\pi$ in  
in any round of
$[t, t + \ell]$ do not yield hierarchical information; 
the length of the gap is $\ell + 1$. The game graph has
\emph{gap size}~$k$ if the length of all gaps in its plays is
uniformly bounded by $k$. 

Clearly, every game graph with finite gap size yields 
recurring hierarchical information. 
Conversely, the automaton construction 
of~Lemma~\ref{lem:automaton-hierarchical-recurring} implies,
via a standard pumping argument, 
that the gap size of any game graph with recurring hierarchical
information is at most the number of states in the
constructed B{\"u}chi automaton.
In conclusion, a game~$G$ yields recurring hierarchical information
if and only if, the size of a gap in any play of~$G$ 
is bounded by $2^{O(n^2|V|^2)}$.

\begin{corollary}\label{cor:gapsize}
If a game yields recurring hierarchical information, then its
gap size is bounded by $2^{O( n^2 |V|^2 )}$, where $n$ is the number of players and $|V|$ is 
the number of positions. 
\end{corollary}

A family of game graphs where 
the gap size grows
 exponentially with the number of positions is illustrated in Figure~\ref{fig:primes}. 
The example is
adapted from~\cite{BerwangerDoyen08}: 
There are two players with no relevant action choices and
they can observe one bit, or a special
symbol that identifies a unique sink position~$v_\bullet$.
The family is formed of graphs $(G_m)_{m \ge 1}$, each 
constructed of $m$ disjoint cycles $(C^r)_{1 \le r \le m}$ of lengths 
$p_1, p_2,\dots, p_m$ corresponding to the first~$m$ prime numbers, respectively.
We number the positions on the cycle corresponding to the $r$-th prime number as
$C^r := \{ c^r_{0}, \dots, c^r_{p_r-1} \}$. 
On each cycle, both players receive the same observation~$0$. 
Additionally, there are two special positions $v_{01}$ and $v_{10}$, 
that yield different observations to the players:
$\beta^1( v_{01} ) = \beta^2( v_{10} ) = 0$ and $\beta^1( v_{01} ) = \beta^2( v_{10} ) = 1$. 
From the initial position~$v_0$ of a game graph~$G_m$, 
Nature can choose a cycle~$C^r$ with $r \le m$. 
From each position $c^r_{\ell} \in C^r$, 
except for the last one with $\ell = p_r-1$, there 
are is a moves to the subsequent cycle position $c^r_{\ell+1}$ and, 
additionally, to $v_{01}$ and $v_{10}$.
In contrast, the last cycle position 
$c^r_{p_r-1}$ 
has only the first position $c^r_{0}$ of the same cycle as a successor. 
From the off-cycle positions $v_{01}$ and $v_{10}$ 
the play proceeds to the unique sink state~$v_\bullet$
that emits the special observation~$\bullet$ to both players.

Now, we can verify, for each game graph~$G_m$, 
that in every play~$\pi$ that proceeds only through cycle positions,
the information sets of the two players are comparable at a prefix history of 
length $t > 2$ in~$\pi$, if and only if, all the first~$m$ primes divide $t-2$; 
any play that leaves a cycle reaches the sink~$v_\bullet$, where 
the  information sets of both players coincide. 
Accordingly, $G_m$ yields recurring hierarchical information. 
On the other hand, since the product of the first $m$ primes 
is exponential in their sum, 
(for a more precise analysis, see \cite{BerwangerDoyen08}), 
we can conclude that the gap size of the game graphs~$G_m$ 
grows exponentially with the number of positions.

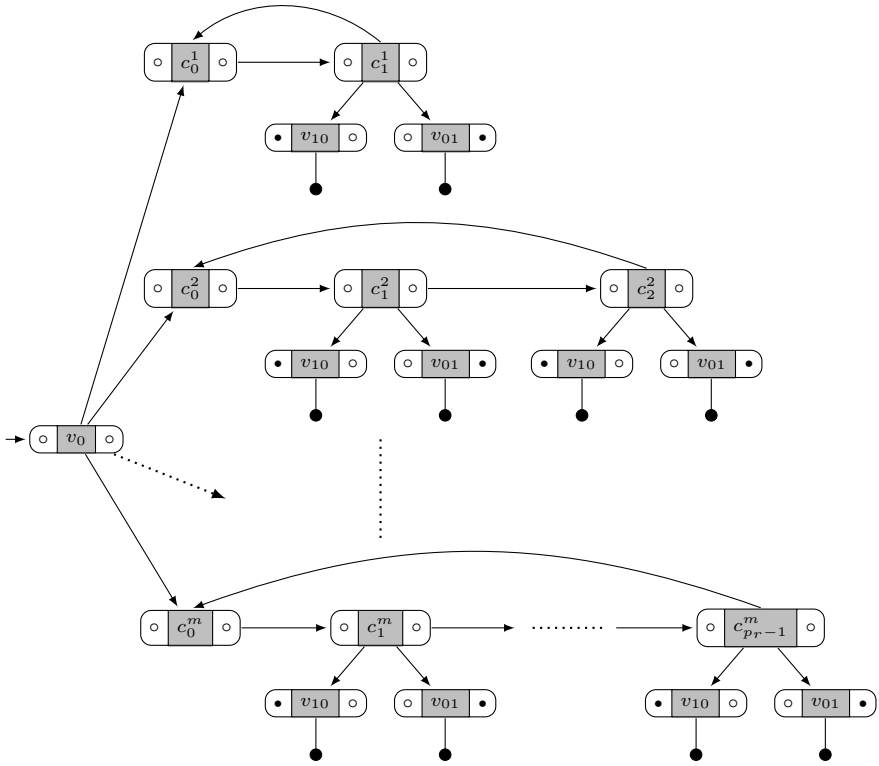
\begin{figure}
\begin{flushleft}
\begin{tikzpicture}

  \node[initial,state3P] 				(0) at (-4,6)         {$\circ$ \nodepart{two} $v_0$ \nodepart{three} $\circ$};


  \node[state3P]              (21) at (-2.5,11) {$\circ$ \nodepart{two} $c_0^1$ \nodepart{three} $\circ$};
  
  \node[state3P]              (22) at (0,11) {$\circ$ \nodepart{two} $c_1^1$ \nodepart{three} $\circ$};
  
  \node[state3P]              (22') at (-.85,10) {$\gob$ \nodepart{two} $v_{10}$ \nodepart{three} $\circ$};
  
  \node[state3P]              (22'') at (.85,10) {$\circ$ \nodepart{two} $v_{01}$ \nodepart{three} $\gob$};
  
  
  \node[state3P]              (31) at (-2.5,8) {$\circ$ \nodepart{two} $c_0^2$ \nodepart{three} $\circ$};
  
  \node[state3P]              (32) at (0,8) {$\circ$ \nodepart{two} $c_1^2$ \nodepart{three} $\circ$};
  
  \node[state3P]              (32') at (-.85,7) {$\gob$ \nodepart{two} $v_{10}$ \nodepart{three} $\circ$};
  
  \node[state3P]              (32'') at (0.85,7) {$\circ$ \nodepart{two} $v_{01}$ \nodepart{three} $\gob$};

  \node[state3P]              (33) at (3.5,8) {$\circ$ \nodepart{two} $c_2^2$ \nodepart{three} $\circ$}; 
  
  \node[state3P]              (33') at (2.65,7) {$\gob$ \nodepart{two} $v_{10}$ \nodepart{three} $\circ$};
  
  \node[state3P]              (33'') at (4.35,7) {$\circ$ \nodepart{two} $v_{01}$ \nodepart{three} $\gob$};

 
 \node[state3P]              (p1) at (-2.5,3.5) {$\circ$ \nodepart{two} $c^m_0$ \nodepart{three} $\circ$};
       
  \node[state3P]              (p2) at (0,3.5) {$\circ$ \nodepart{two} $c^m_1$ \nodepart{three} $\circ$};
  
   \node[state3P]              (p2') at (-0.85,2.5) {$\gob$ \nodepart{two} $v_{10}$ \nodepart{three} $\circ$};
  
  \node[state3P]              (p2'') at (0.85,2.5) {$\circ$ \nodepart{two} $v_{01}$ \nodepart{three} $\gob$};

 \node[state3P]              (p3) at (5,3.5) {$\circ$ \nodepart{two} $c^m_{p_r-1}$ \nodepart{three} $\circ$}; 
 
  \node[state3P]              (p3') at (4.15,2.5) {$\gob$ \nodepart{two} $v_{10}$ \nodepart{three} $\circ$};
  
  \node[state3P]              (p3'') at (5.85,2.5) {$\circ$ \nodepart{two} $v_{01}$ \nodepart{three} $\gob$};
 

 \path (0) edge (21) edge (31) edge (p1);

\path (21) edge (22);
\path (22) edge (22') edge (22'');


\path (22'.south) edge[-*] (-.85,9.2);

\path (22''.south) edge[-*] (.85,9.2);

\path (22.north) edge[bend right=40] (21.north);

\path (31) edge (32);
\path (32)  edge (32') edge (32'') edge (33);
\path (33) edge (33') edge (33'');


\path (32'.south) edge[-*] (-.85,6.2);

\path (32''.south) edge[-*] (0.85,6.2);

\path (33'.south) edge[-*] (2.65,6.2);

\path (33''.south) edge[-*] (4.35,6.2);

\path (33.north) edge[bend right=20] (31.north);


\path (p2'.south) edge[-*] (-.85,1.7);

\path (p2''.south) edge[-*] (0.85,1.7);

\path (p3'.south) edge[-*] (4.15,1.7);

\path (p3''.south) edge[-*] (5.85,1.7);

\path (p3.north) edge[bend right=20] (p1.north);

\path (p1) edge (p2);
\path (p2) edge (1.8,3.5) edge (p2') edge (p2'');
\path (p3) edge (p3') edge (p3'');

 \draw[->,>=latex]  (3.1,3.5) --  (p3);
 \draw[-,dotted,thick] (2,3.5) -- (3, 3.5);
 \draw[-,dotted,thick] (0,6) -- (0, 4.6);
 \draw[->,>=latex, dotted,thick]  (0) --  (-2,5.2); 

\end{tikzpicture}
\end{flushleft}
 
\caption{Game with an exponential gap of non-hierarchical information (for better readability, 
the positions $v_{01}$, $v_{10}$, and $v_\bullet$ are multiply represented)}
 \label{fig:primes}
\end{figure}

As a further consequence of the automaton construction
in Lemma~\ref{lem:automaton-hierarchical-recurring}, it follows that 
we can decide whether a game graph~$G$ yields recurring hierarchical
information, by constructing the corresponding B\"uchi automaton and
checking whether it accepts all plays in~$G$. 
However, due to the exponential 
blow-up in the determinisation of the automaton, 
this straightforward approach would 
require exponential time (and space) 
in the size of the game graph and the number of (pairs of) players.
Here, we describe an on-the fly 
procedure that yields better complexity bounds.

\begin{theorem}\label{thm:deciding-hierarchical-recurring}
The problem of deciding whether a game graph yields recurring hierarchical
information is \textsc{PSpace}-complete.
\end{theorem}

\begin{proof}
We describe a nondeterministic procedure for verifying that
an input game~$G$ does \emph{not} yield recurring hierarchical information, 
that is, there exists a play in $G$ such that from some round $t$ onwards, 
no prefix of length $\ell \ge t$ yields hierarchical information.
By looking at the deterministic B{\"u}chi automaton~$B$ 
constructed in Lemma~\ref{lem:automaton-hierarchical-recurring}, 
we can tell that this is the case if, and only if, 
there exists a finally periodic play $\pi = \tau \rho^\omega \in V^\omega$ 
such that the run of~$B$ on~$\pi$
visits only non-accepting states after reading the 
prefix~$\tau$ and, moreover, it returns to a previously visited 
state when, but not earlier than, reaching the
prefix $\tau \rho$. 
In other words, the run on~$\tau\rho$
induces a \emph{lasso} in~$B$, 
and hence, the length of $\tau\rho$ is bounded 
by the number of states~$|V|\,2^{2 n^2|V|^2}$ in the automaton.

The idea of the procedure, pictured in Algorithm~\ref{algo:decision-rec}, 
is to guess such a history~$\tau\rho$, and to keep track of the 
states visited in the corresponding run of the automaton $B$, or more precisely, 
in the powerset construction of the automaton~$A$ 
from Lemma~\ref{lem:automaton-hierarchical-history}. 

Let us first fix a pair of players $i, j$.
Then, every state reachable by~$A$ upon reading the prefix~$\pi$ 
of an input word from $V^\omega$ is described by a tuple $(u, c, w, d)$ 
consisting of a pair of game positions~$u,w$ and two binary flags $c,d$ 
such that there exist histories $\pi', \pi''$ in $G$ that end at $u$ and $w$, 
satisfying $\pi \sim^i \pi'$, $\pi \sim^j \pi''$, and the flags 
$c$ or $d$ are set if, and only if, $\pi'' \not\sim^i \pi$ 
or $\pi \not\sim^j \pi'$, respectively.
Essentially,~this means that 
$\pi', \pi''$ are candidates for witnessing that players~$i,j$ have
incomparable information at some continuation of~$\pi$. 
Every state in the automaton~$B$ on the powerset of~$A$ 
records a set~$Z_{ij} \in (V \times \{ 0, 1\})^2$ of such tuples, 
each collecting the terminal positions of all witness candidates for $i$, $j$ 
flagged correspondingly -- we call 
such a set a \emph{cell}. 
At the beginning, 
the procedure guesses one cell~$\hat{Z}_{ij}$ for every pair of players $i,j$ (Line~\ref{alg:target}); 
we call such a collection of cells a \emph{configuration}. 
The guessed configuration stands for a state of~$B$ 
that shall be reached after reading 
$\tau$ and to which the run returns at $\tau\rho$. 
Then, starting from the initial configuration, where all players 
just see the initial position $v_0$ (Line~\ref{alg:init}), the  
procedure generates successively game positions of $\tau$ while updating 
the current configuration to simulate the run of~$B$ on~$\tau$.
The current configuration~$Z$ 
summarises the possible runs of~$A$ on the 
prefix of~$\tau$ generated so far, 
and it is updated for every new game position~$v$
according to Procedure~\ref{alg:update}.
Once the configuration~$Z = \hat{Z}$ is reached, 
the procedure enters a new loop. 
Here, it verifies in every iteration that the current configuration 
indeed contains a witness of incomparability for the input history, 
that is, there exists, in the cell of some pair of players,
a state of~$A$ in which both distinguishability flags are set (Line~\ref{alg:if-hier}). 
Provided the test succeeds, 
the procedure successively guesses the positions of~$\rho$ while updating the current 
configuration,  
until~$\hat{Z}$ is reached again, in which case the procedure accepts.
Apart of the case when the test in Line~\ref{alg:if-hier} fails, 
the procedure also rejects by looping.

\paragraph{Correctness and soundness} 
The procedure mainly requires space to store the configurations~$Z, \hat{Z}$ 
that collect a set of tuples from $(V \times \{0, 1 \})^2$
for each pair of players, hence it runs in polynomial space. 
To show correctness, we consider the sequence~$\pi$ of positions~$v$ 
generated in Lines~\ref{alg:next-prefix} and~\ref{alg:next-cycle}, and  
argue that it forms a history in~$G$ and that, at every iteration of the 
loops in Lines~\ref{alg:prefix} and~\ref{alg:cycle}, the configuration~$Z$ 
contains, in each cell~$Z_{ij}$, precisely the  
set of pairs~$u, w$ of terminal positions 
reachable by candidate witnesses $\pi' \sim^i \pi$ such that
$\pi'' \sim^j \pi$, with associated flags $c,d$ indicating correctly whether
$\pi' \sim^j \pi''$ and 
$\pi'' \sim^i \pi$, due to the way they are maintained in
 Lines~\ref{alg:flag-i} and~\ref{alg:flag-j} of Procedure~\ref{alg:update}. 
Accordingly, if the procedure accepts, 
then there exists a finally periodic play~$\pi := \tau \rho^\omega$ in~$G$ 
such that the information sets of at least two players are incomparable, 
in every round from~$\tau$ onwards. 
Soundness follows from the construction of the B{\"u}chi automaton in 
Lemma~\ref{lem:automaton-hierarchical-recurring} 
and the observation that every run of the procdedure corresponds to a 
run of the automaton with the same acceptance status. 

\paragraph{Hardness} 
Finally, we argue that the problem of deciding whether a game graph yields 
recurring hierarchical information is $\textsc{PSpace}$-hard by 
reduction from the universality problem for nondeterministic finite automata, 
shown by~\cite{MeyerS72} to be $\textsc{PSpace}$-hard. 
Given a nondeterministic automaton~$A = (Q, \Sigma, \Delta, q_0, F)$ 
over an alphabet~$\Sigma$, we 
construct a game graph~$G$ for two players 
with no action choices 
and with observations
in $B^1 = B^2 = \Sigma\times\{0,1\}$ corresponding to input letters for~$A$ tagged with one bit.
The set of positions in~$G$ contains $q_0$ and all letter-state pairs 
$(a, q) \in \Sigma\times Q$. From every state $(a, q)$ and from $q=q_0$, there is a move 
to position $(a', q')$ if $(q, a', q') \in \Delta$. Each position $(a,q)$ 
yields the same observation $(a, 0)$ to both players (the observation at $q_0$ is irrelevant).
Thus, every run of~$A$ on a word $\alpha \in \Sigma^*$ corresponds to a 
unique history in $G$ where both
players observe the letters of $\alpha$ tagged with $0$.  
In addition, $G$ has two fresh positions $v_a$ and $v_a'$, for every letter $a \in A$, with 
observations $\beta^1( v_a ) = \beta^2( v_a') = (a, 1)$ 
whereas $\beta^1( v_a' ) = \beta^2( v_a) = (a, 0)$. Whenever there is a move in~$G$ 
from a position $v$ to some position $(a, q)$ with $q \in F$, we also allow a move from $v$ to $v_a$. 
These fresh positions have one common successor, 
identified by a distinct observations to both players (essentially, indicating that the game is over). 
Clearly,~$G$ can be constructed from~$A$ in polynomial time.
  
Now, consider a nontrivial word~$\alpha \in \Sigma^*$ and suppose that it 
admits an accepting run in~$A$. 
At the corresponding history~$\pi$ in $G$, which yields
the letters of ~$\alpha$ tagged with~$1$ as an observation to both players, 
the information sets are incomparable, because
each player considers it possible that the other 
recieved the last letter with a~$0$-tag. 
In contrast, if~$\alpha$ is rejected, 
the information sets at the corresponding history in $G$ coincide, 
hence we have hierarchical information. 
In conclusion, the language of the automaton~$A$ is universal 
if, and only if, the constructed game graph yields herarchical information.
 \qed
\end{proof}

\SetKw{KwAccept}{accept}
\SetKw{KwReject}{reject}
\SetKw{KwGuess}{guess}
\SetKwArray{KwType}{type} 
\SetKwArray{KwVar}{var} 
\begin{algorithm}[t]
\caption{Deciding recurring hierarchical information}
\label{algo:decision-rec}
{
 \KwData{game graph $G = (V, E, \beta)$ for $n$ players}
 \KwResult{accept if $G$ does not yield recurring hierarchical information}
 \BlankLine
 \KwType $\textit{cell}$\,:  \text{subset of} $(V \times \{0,1\})^2$\;
 \KwType $\textit{configuration}$\,:  matrix of $\textit{cells}$ $(Z_{i,j})_{1\le i < j \le n}$\;
 \KwVar $Z, \hat{Z}$\,: \text{configurations}\;
 \KwVar $v, \hat{v}$\,: \text{positions in} $V$\;
 \KwVar $i, j$\,: \text{players in} $\{1, \dots, n \}$\;
 \BlankLine
 \nl \KwGuess{$(\hat{v}, \hat{Z})$\label{alg:target}} \tcp*[r]{fix target on cycle}
 \nl $v \gets v_0$\;
 \nl \lForEach(\tcp*[f]{initial configuration}){$i, j$ with 
   $i < j$}{$Z_{i,j} \gets \{ (v_0, 0, v_0, 0 )\}$\label{alg:init}}     
 \nl \While{$(v, Z) \neq (\hat{v}, \hat{Z})$\label{alg:prefix}}
     {
       \nl \KwGuess{$v \in vEA$\label{alg:next-prefix}}
       \tcp*[r]{guess next history state}
       \nl $Z \gets \mathrm{Update}(v, Z)$
       \tcp*[r]{follow powerset construction}
     }
     \nl 
     \Repeat{$(v, Z) = (\hat{v}, \hat{Z})$\label{alg:cycle}}{
       \nl \lIf(\tcp*[f]{hierarchical information}){$\displaystyle{ 
           \bigwedge_{i, j} Z_{i,j} \cap(V \times \{1\})^2 = \emptyset}$ 
         \label{alg:if-hier}} {
         \KwReject
       }
       \nl \KwGuess{$v \in vEA$\label{alg:next-cycle}}\;
       \nl $Z \gets \mathrm{Update}(v, Z)$\;
     }(\tcp*[f]{cycle found})
     \nl \KwAccept \;        
}
\end{algorithm}

\begin{procedure}[t]
\caption{Update($v, Z$)}
\label{alg:update}
{
 \KwData{new position~$v$, current configuration~$Z$}
 \KwResult{successor configuration after observing $\beta(v)$}
     \nl \ForEach{$i,j$}{
       \nl \ForEach{$(u,c,w,d) \in Z_{i,j}$}
               {\nl \ForEach{$u' \in uEA, w' \in wEA$ with $\beta^i( u' ) = \beta^i ( v )$ and 
                   $\beta^j( w' ) = \beta^j ( v )$}
               {
                 \nl $c' \gets c \lor \beta^j( u' ) \neq \beta^j( w' )$ \label{alg:flag-i}
                 \tcp*[r]{flag witness for $j \preceq i$}
                 \nl $d' \gets d \lor \beta^i( u' ) \neq \beta^i( w' )$ \label{alg:flag-j}
                 \tcp*[r]{flag witness for $i \preceq j$}
                 \nl $Z'_{i,j} \gets Z'_{i,j} \cup \{ ( u', c', w', d') \} $\;
               }
               }\;
     }
   }\KwRet $Z'$\;

\end{procedure}

We can show that the synthesis problem for the class of games 
with recurring hierarchical information is finite state-solvable, 
at least in the case when the winning 
conditions are observable. We conjecture that the result extends to the 
general case.

\begin{theorem} 
For games with recurring hierarchical information and observable
$\omega$-regular winning conditions, the synthesis problem is 
finite-state solvable.
\end{theorem}

\begin{proof}
The argument relies on the tracking construction 
decribed in~\cite{BKP11}, which reduces the problem of solving 
distributed games with imperfect information for $n$ players against
Nature to that of solving a zero-sum game for two players 
with perfect information. 
The construction proceeds via
an unravelling process 
that generates epistemic models of the player's information along the
rounds of a play,
and thus encapsulates their uncertainty. 

This process
described as ``epistemic unfolding''
in the paper \cite[Section~3]{BKP11} is outlined as follows.
An \emph{epistemic model} for a game graph $G$ with the usual
notation, is a 
Kripke structure $\calK = (K, (Q_v)_{v \in V}, (\sim^i)_{1 \le i \le n})$ 
over a set $K$ of histories of the same length in in $G$, 
equipped with predicates $Q_v$ designating the histories that end in position
$v \in V$ and with 
the players' indistinguishability relations $\sim^i$.
The construction keeps track of how the knowledge 
of players about the actual history is updated during a round,
by generating for each epistemic model $\calK$
a set of new models, one for 
each assignment of an action profile~$a_k$ 
to each history $k \in K$ such that  
the action assigned to any player~$i$  
is compatible with his information, i.e.
for all $k, k' \in K$ with
$k \sim^i k'$, we have $a_k^i = a_{k'}^i$. 
The update of a model~$\calK$ with such an action 
assignment~$(a_k)_{k\in K}$ leads to a new, 
possibly disconnected epistemic model $\calK'$ 
over the universe 
\begin{align*}
  K' = \{k a_k w \mid k \in K \cap Q_v \text{ and } (v,a_k,w) \in
  E \},
\end{align*} 
with predicates $Q_w$ designating the histories $k a_k w \in K'$,
and with $k a_k w \sim^i k' a_k w'$ whenever
$k \sim^i k'$ in~$\calK$ and $w \sim^i w'$ in~$G$.
By taking the connected components of this updated model under 
the coarsening $\sim := \bigcup_{i = 1}^{n}\!\!\sim^i$, 
we obtain the 
set of epistemic successor models of $\calK$ in the unfolding. 
The tracking construction starts from 
the trivial model that consists only of the initial 
position of the game~$G$. 
By successively applying the update, it unfolds a tree 
labelled with  
epistemic models, which corresponds to a game graph $G'$ for two players 
with perfect 
information
where the strategies of one player translate into 
distributed strategies in~$G$
and vice versa. According to~\cite[Theorem 5]{BKP11}, for any winning condition,
a strategy in~$\calG'$ is winning if and
only if the corresponding 
joint strategy in~$\calG$ is so.
 
The construction can be exploited algorithmically if the
perfect-information tracking of a
game can be folded back into a finite game. 
A homomorphism from an epistemic model 
$\calK$ to $\calK'$ is a function $f: K \to K'$ that preserves the
state predicates and the indistinguishability relations, that is, 
$Q_v(k) \Rightarrow Q_v (f(k)) $ and 
$k \sim^i k' \Rightarrow f(k) \sim^i f(k')$. 
The main result of~\cite{BKP11} shows that, whenever two nodes of the unfolded 
tree carry homomorphically equivalent labels, 
they can be identified without changing the 
(winning or losing) status of the game~\cite[Theorem 9]{BKP11}. 
This holds for all imperfect-information games 
with $\omega$-regular winning conditions that are observable. 
Consequently, the strategy synthesis problem is decidable for a
class of such games, whenever  
the unravelling process of any game in the class 
is guaranteed to generate only finitely many
epistemic models,  
up to homomorphic equivalence.

Game graphs with recurring hierarchical information satisfy this
condition. Firstly, for a fixed game, there exist only finitely many
epistemic models, up to homomorphic equivalence, 
where the $\sim^i$-relations are totally ordered by
inclusion~\cite[Section 5]{BKP11}. In other words,
epistemic models of bounded size are sufficient to describe 
all histories with hierarchical information. 
Secondly, 
by~Corollary~\ref{cor:gapsize}, from any history with hierarchical information,
the (finitely branching) tree of continuation histories with incomparable information is of
bounded depth, hence only  finitely many epistemic models can occur in
the unravelling. 
Overall, this implies that every game with recurring hierarchical
information and observable winning condition has a finite quotient
under homomorphic equivalence. 
According to~\cite[Theorems 9 and 11]{BKP11}, we can conclude that the distributed strategy
problem for the class is finite-state solvable.

\qed
\end{proof}

We point out that 
the number of hierarchic 
epistemic models in games, and thus the complexity of our synthesis procedure,
grows nonelementarily with the number of players. 
This should not come as a surprise, as the solution
applies in particular to games with hierarchical observation 
under safety or reachability conditions (here, the distinction 
between observable and non-observable conditions  
is insubstantial), and it is known that already in this case
no elementary solution exists (see \cite{PetersonRei79,AzharPetRei01}).

\section{Games and Architectures}

The game perspective focuses on the 
interaction between players letting their individual entity 
slide into the background. 
For instance, two players may interact locally without
observing actions of a third, more remote player. 
Or, there may be independent teams specialised in solving particular 
tasks independently in response to inputs received from a central authority.
Such organisational aspects can be crucial for 
coordinating distributed processes, 
yet they are hardly apparent from the representation of the system as
a game graph.

\subsection{Monitored architectures}

To interpret our results with regard to processes and the 
infrastructure in which they interact, 
we wish to illustrate how the game model relates to the 
standard framework of distributed architectures of 
reactive systems introduced 
by \cite{PnueliRos90}. 
Here, we formulate the setting in more general terms 
to allow for a meaningful translation of games into architectures, 
and for the presentation of new decidable architectures.  
In our framework, 
processes are equipped with a local 
transition structure according to which actions are enabled or disabled 
depending on previous actions and observations. 
Most importantly, the communication structure is not hard-wired as in the classical setting, 
but instead implemented via a
finite-state device that can dispatch information 
in different ways depending on the previous run. 

A \emph{process} is represented by an 
automaton $\calP :=(Q, A, B, q_0, \delta)$ on a set~$Q$ of states, 
with output alphabet~$A$, input alphabet~$B$,
initial state~$q_0 \in Q$, 
and a partial transition function $\delta: Q \times A \times B \to Q$. 
The symbols of the output alphabet~$A$ are called actions
and those of the input alphabet~$B$ observations. 
We say that an action~$a \in A$ is \emph{enabled} at a state~$q \in Q$, 
if $(q, a, b) \in \mathrm{dom}( \delta )$ for some~$b \in B$.  
We assume that for every state~$q \in Q$ 
the set $\mathrm{act}( q)$ of enabled actions is nonempty, 
and that for each enabled action~$a \in \mathrm{act}( q)$ 
and each observation~$b \in B$, the transition $\delta(q, a, b)$ is defined.

The automaton describes the possible behaviour of the process as follows:
Every run starts in the initial state $q_0$. In any state~$q$, 
the process chooses one of the available actions~$a \in \mathrm{act}( q )$, 
then it receives an observation $b \in B$ and 
switches into state~$\delta( q, a, b)$ to proceed.
From a local perspective, the automaton expresses
which actions of the process are enabled or disabled, 
depending on the sequence of previous actions 
and observations.
From an external perspective,~$\calP$ defines a (local) \emph{behaviour} relation
$R_\calP \subseteq (A \times B)^\omega$ consisting of the sequences of action-observation 
symbols in possible runs, which 
represents
how the process may act in response to the observation sequence received as 
input.\footnote{Such automata are called 
synchronous sequential transducers or nondeterministic 
generalised sequential machines in the literature.} 
In the following, we will not distinguish between 
the automaton~$\calP$ and the defined behaviour relation~$R_\calP$.  

A \emph{black-box} process is one with a single state, i.e., 
the set of enabled actions is independent of previous actions or observations.   
A \emph{white-box} process is one where every state has precisely one enabled action -- 
hence, a Moore machine, or a finite-state strategy 
in the game setting.  
A \emph{program}, or strategy for process~$\calP$ is a white-box $\calS$ 
with behaviour $\calS \subseteq \calP$. 


A \emph{monitored architecture} is represented by two elements: 
a collection of processes~$\calP^1, \dots, \calP^n$ plus 
a distinguished black-box process $\calP^0$ called \emph{Environment} on the one hand,  
and a specification of the communication infrastructure, 
called \emph{view monitor}, on the other hand.
When we refer to process~$i$, 
we identify all associated  elements with a superscript and write
$\calP^i = (Q^i, A^i, B^i, q_0^i, \delta^i)$.
The set of \emph{global actions} is the 
product $\Gamma := A^0 \times A^1 \times \dots \times A^n$ of the 
action sets of all processes, including the Environment.
Then, a \emph{view monitor} is a Mealy machine 
$\calM = (M, \Gamma, B, m_0, \mu, \nu)$ 
with input alphabet~$\Gamma$ and an output alphabet 
consisting of the product 
$B := B^1 \times \dots \times B^n$ 
of the observation alphabets of all processes.

Hence, the monitor~$\calM$
transforms sequences of global actions into a tuple of   
observations, one for each process~$i \in \{ 1, \dots n \}$. 
The observations of the Environment are irrelevant, we always assume that~$B^0 = \{ 0 \}$. 
Since the Environement is a black box, it 
is completely specified by its set of actions, which 
already appears in the description of the view monitor.  
Therefore, $(\calP^1, \dots, \calP^n, \calM)$ 
yields a complete description of a monitored architecture. 


A \emph{global behaviour} in a monitored architecture is an 
infinite sequence of global actions and observations
$\pi := (a_0, b_0) (a_1, b_1), \dots  \in (\Gamma \times B)^\omega$
such that the observations corresponds to the 
output of the view monitor, $b_t = \mu( a_0, \dots a_t )$ for all~$t\ge0$ and,  
for every process~$i$, the corresponding action-observation sequence $(a^i_0, b^i_0)(a^i_1, b^i_1) \dots$ 
represents a behaviour in $\calP^i$.   
We refer to the sequence of $a_0 a_1 \dots$ of global actions in a global behaviour
as a (global) \emph{run}.

A \emph{distributed program} is a collection 
$\calS = (S^1, \dots, S^n)$ of programs, one for every process. 
A global behaviour~$\pi$ is \emph{generated} by a distributed program~$\calS$, 
if for every process~$i$, 
the action-observation sequence $(a^i_0, b^i_0)(a^i_1, b^i_1) \dots$ represents a behaviour 
in~$\calS^i \subseteq \calP^i$.


The \emph{run tree} of a distributed program~$\calS$ is the set 
$T_\calS \subseteq \Gamma^*$ 
of prefixes of global runs generated by~$\calS$. 
Properties of runs are described by a linear-time or branching-time specification, 
given by an $\omega$-word automaton or an $\omega$-tree automaton over~$\Gamma$, respectively. 
A run tree~$T$ satisfies a linear-time specification if all branches are accepted, 
and it satisfies a branching-time specification, if the tree $T$ is accepted by the 
specification automaton.
We say that a distributed program~$\calS$ is \emph{correct} with respect to a 
specification 
if the generated run tree~$T_\calS$ satisfies the specification.
Given an architecture together with a specification~$\Phi$, 
the distributed synthesis problem asks whether there exists a distributed
program~$\calS$ that is correct with respect to the specification~$\Phi$.


\subsection{From architectures to games and back}

A monitored architecture $(\calP^1, \dots \calP^n, \calM)$, can be  
translated into a distributed game $G = (V, E, \beta)$ for $n$ players as follows. 
The set~$V$ of positions is 
the product 
$B \times M \times Q^1 \times \dots \times Q^n$ of the global observation space
with the state sets of the view monitor and of all processes 
(excluding the Environment), with initial position 
$(b_0, m_0, q_0^1, \dots, q_0^n)$ for some (irrelevant) initial observation.
There is a move $((b, m, q), a, (b', m', q')) \in E$, 
whenever the components are updated correctly, that is, 
the global observation $b' = \nu(m, a)$, 
the memory state of the view monitor $m' = \mu(m, a)$, 
and the local control state $q'^i := \delta( q^i, a^i, b'^i )$ 
for every process~$i$.
(Notice that the observation at the source does not matter). 
Finally, the observation function for each player~$i$ assigns 
$\beta^i(b, m, q) := b^i$.

Every prefix of a global run in the architecture induces a 
unique run prefix in the view monitor~$\calM$ and thus a 
history in the associated game, such that any program~$s^i$ for a process~$i$ 
corresponds to a strategy for player~$i$ and every specification induces 
a winning condition expressible by an automaton over the states of~$\calM$. 
Further, the outcome of a 
distributed finite-state strategy induces the run tree generated 
by the corresponding distributed program, hence
every distributed programs that satisfies the specification   
correspond to winning strategy, and vice versa. 

In this paper, we only considered games with linear-time winning conditions.
Nevertheless, the automata technique of \cite{KupfermanVar01} 
for solving games with hierarchical observation (Theorem~\ref{thm:KV01}) 
also applies to the branching-time setting. 
Thus, our decidability results for static 
and dynamic hierarchical information 
(Theorem~\ref{thm:static-decidable} and~\ref{thm:dynamic-decidable})
generalise immediately to branching-time winning conditions. 

In summary, we can regard system architectures as a syntactic variant of 
games, modulo the transformation described above.

\begin{proposition}\label{prop:architecture-to-game}
Every instance of the distributed synthesis problem over architectures 
can be reduced to an instance over games
such that the finite-state solutions are preserved. 
For an architecture $(\calP^1, \dots, \calP^n, \calM)$ and a specification given by a parity 
automaton~$A_\Phi$, 
the reduction runs in time  $O(\, |\calM| \,( | A_\Phi | + |\calP^1 \times \dots \times \calP^n|)\,)$.  
\end{proposition}

For the converse, to translate a given game graph $G = (V, E, \beta)$ for $n$ players 
into an architecture, we proceed as follows. For every player~$i$, we first consider
the localised variant $G^i = (V^i, E^i, \beta^i)$ of the game graph, obtained by replacing
every action label~$a$ in a move $(u, a, v) \in E$ with the component $a^i$, 
and retaining only the observation function $\beta^i$. 
Then, we view the resulting one-player game graph as a 
nondeterministic finite-state automaton over the alphabet~$A^i \times B^i$,   
by placing the observation symbol~$b^i = \beta^i( v )$ from the target~$v$ 
of any move $(u, a^i, v)$ on the 
incoming edge to yield the transition $(u, (a^i, b^i), v)$. 
Finally we determinise and minimise this automaton. The resulting automaton
$(Q, A^i\times B^i, v_0, \delta)$
may not be complete in the second input component, as it is required for a 
process. Therefore, we add a fresh sink~$z$ with incoming transitions 
$(q, (a^i, b^i), z)$ from any state~$q$ where action $a^i$ is enabled, but 
$\delta(q, (a^i, b^i))$ is undefined. 
The automaton~$\calP^i$ obtained in this way corresponds to a process.

To construct the associated view monitor, we expand the alphabet~$A$ 
of joint actions by adding a set~$A^0$ of Environment actions, 
which we call \emph{directions}.  
Intuitively, directions correspond to choices of Nature. 
Accordingly, we choose $A^0$ to be of size outdegree of~$G$ and  
expand the action profiles in the moves of $E$ with directions, 
such that at any position, if there are two different outgoing moves, 
the direction in their action labels differs. 
Thus, we obtain a representation of~$G$ as a deterministic automaton with 
transition relation $\mu: V \times \Gamma \to V$ over the 
expanded alphabet $\Gamma := A^0 \times A^1 \times \dots \times A^n$. 
Additionally, we define an output function
$\nu: V \times \Gamma \to B$ that assigns to every 
position~$u$ and every direction-action profile $a \in \Gamma$, the profile of
observations $\beta^i( v )$ for $v = \mu( u, a)$. 
Finally,  we consider the
Mealy automaton $\calM = (V, \Gamma, B, v_0, \mu, \nu)$ as a 
view monitor for the collection of processes~$\calP^i$. 

By way of this translation, the histories in~$G$ correspond 
to prefixes of runs in the architecture~$(\calP^1, \dots, \calP^n, \calM)$, 
such that any finite-state strategy for a player~$i$ corresponds to 
a program for process~$i$, and the outcome of 
any finite-state strategy profile~$s$
corresponds to the run tree generated by the distributed program corresponding to~$s$. 
Since the game graph is deterministic in the action alphabet expanded with directions, 
every winning condition induces a specification, such that 
the solutions to the distributed synthesis problem are preserved. 
The specification is obtained as a product of the winning-condition automaton  
synchronised  with the  game graph (in the determinised variant with direction labels); 
in addition, any run that reaches a sink is deemed correct. 
Accordingly, the specification automaton is polynomial in the size of the game.
However, as the process automata are constructed by determinisation and minimisation 
of the localised 
game graph, the overall translation
involves a necessary exponential blowup (see \cite{vanGlabbeek2008}).

\begin{proposition}\label{prop:game-to-architecture}
Every instance of the distributed synthesis problem over games
can be reduced, in exponential time, to an instance over architectures
such that the finite-state solutions are preserved.
\end{proposition}

\subsection{Pipelines, forks, and hierarchical observation}

The classical setting of \cite{PnueliRos90} corresponds to the special case of 
architectures where all processes $\calP^1, \dots, \calP^n$ are black boxes, 
and the action alphabet $A^i$ of each process~$i$ consists of a product 
$X^i \times Y^i$ of sets of \emph{control} and \emph{communication} signals.
The input alphabets and the communication infrastructure are described 
by a directed graph $((\{ 1, \dots, n \}, \mathrm{\to})$ 
with arcs~$i\to j$ expressing that, in every round, 
the communication signal $y^i$ emitted with the
action~$a^i=(x^i, y^i)$ of process~$i$ is received by process~$j$. 
Thus, each player~$i$ has an observation alphabet $B^i$ formed by 
the product of the signal sets $Y^j$ of all its predecessors $j \to i$ in the graph. 
In our framework, the communication graph corresponds to 
a view monitor~$\calM^i$ with only one state, which works as follows: upon input of
a global action $((x^0, y^0), (x^1, y^1) \dots, (x^n,y^n))$, output 
the global observation~$b$ composed of the collection of signals $b^i:=(y^j|~j \to i)$ 
for each player~$i \in \{1, \dots, n\}$.


A \emph{pipeline} architecture in the Pnueli-Rosner setting 
is one where the communication graph 
is a directed path $0 \to 1 \to 2 \dots \to n$. 
The main result of~\cite{PnueliRos90}, later extended by \cite{KupfermanVar01},
shows that the distributed synthesis problem is solvable for 
pipelines architectures. 
We prove that 
games with hierarchical observation and regular winning conditions 
can be translated into  
pipeline architectures with regular specifications
such that the solutions to the distributed synthesis problem are preserved, 
thus justifying our statement in Theorem~\ref{thm:KV01} as a 
corollary of the cited results.


\begin{theorem}\label{thm:hierarchical-obs-pipeline}
The distributed synthesis problem for games with hierarchical observations 
reduces to the corresponding problem for pipeline architectures.  
\end{theorem}

\begin{proof}
Consider a game graph $G = (V, E, \beta)$ that yields hierarchical observation
with a winning condition~$W \subseteq V^\omega$ given by an automaton. 
For simplicity, let us assume that the information hierarchy among the~$n$ 
players follows the nominal order $1 \le \dots \le n$. 
Accordingly, the observation of each player~$i > 1$
is determined by the observation of player~$i-1$, in the sense that 
there exists a function $f^i: B^{i-1} \to B^i$ 
such that $\beta^i( v ) = f^i( \beta^{i-1} ( v ))$
for every position $v \in V$.
By defining $f^1: B^1 \to B^1$ to be the identity,
we can hence write $\beta^i( v ) = f^{i} \circ f^{i-1} \circ \dots \circ f^1 ( \beta^1( v ) )$.

First, we transform the given game instance  
into a monitored architecture $(\calP^1, \dots, \calP^n, \calM)$ and a 
specification~$\Phi$ constructed as in 
Proposition~\ref{prop:game-to-architecture}. 
As the game yields hierarchical information, 
the output function of the view monitor~$\calM$ satisfies
$\nu^i( m, a ) = f^{i} \circ f^{i-1} \circ \dots \circ f^1 ( \nu^1( m, a ) )$, 
for each player~$i$, every memory state $m$, 
and global action $a$. 
Recall that~$\calM$ is
deterministic in the first component of its input alphabet, which 
corresponds to Environment actions (directions). Hence, we can 
write $\nu^1( m, a)$ as $\nu^1( m, a^0 )$.

Next, we sequentialise the obtained architecture:
For each process~$i < n$, we expand the action alphabet~$A^i$ 
with the observation alphabet~$B^{i+1}$ of the next player. 
Thus, we set 
$\hat{A}^i := A^i \times B^{i+1}$ and 
consider the actions in~$A^i$ as internal signals and the 
observations~$B^{i + 1}$ as communication signals. 
The actions of the last process
carry no relevant communication symbols, 
we set~$\hat{A}^n = A^n \times \{ 0 \}$. 
Now, we modify each process~$i < n$ to output in addition to his action $a^i$ 
(now an internal signal), the observation value $f^{i+1}( b^i )$ as
a communication symbol -- where~$b^i$ refers to the observation 
received by process~$i$ in the previous period (or the default initial observation); 
to implement this, the value $f^{i+1}( b^i )$ is stored in the state of the process
when receiving $b^i$ and then added to all outgoing transitions. 
Let us call the modified processes~$\hat{\calP}^i$. 
Again, the last process remains unchanged. 

Notice 
that process~$\hat{\calP}^1$ 
relies on observations $\nu^1( m, a^0 )$ 
rather than just environment action $a^0$. To fix this, we 
take its synchronised product with $\calM$. The resulting process 
$\calM \times \hat{\calP}^1$  
inputs only environment actions and updates the 
monitor state internally, to yield the same output $\nu^1(m, a^0)$ as 
$\hat{\calP}^1$.

As a view monitor~$\hat{\calM}$, we take the standard monitor for the communication graph 
$0 \to 1 \to 2 \to \dots \to n$, in the sense of \cite{PnueliRos90}, 
as described in the beginning of the subsection. 

To account for the latency introduced by sequentialisation,
we adjust the specification~$\Phi$ over the original set of global actions 
$\Gamma = A^0 \times A^1 \times \dots A^n$. 
For every infinite word $\alpha = a_0 a_1 \dots \in \Gamma^\omega$ 
formed of global actions $a_t = (a^0_t a^1_t \dots a^n_t)$ for all~$t \ge 0$, we denote
by $\mathrm{pipe} ( \alpha ) := a'_0, a'_2, \dots \in \Gamma^\omega$ 
the infinite sequence of global actions 
where $a'_t = (a_t^0, a_{t+1}^1, \dots, a_{t+n}^n)$, for all~$t > 0$. 
We can verify that for every global run $\alpha$ in $(\calP, \calM)$, 
the sequence $\mathrm{pipe} ( \alpha )$ represents a run in the 
pipeline $(\hat{\calP}, \hat{\calM})$. Finally, we define the
specification $\hat{\Phi} := \{ \mathrm{pipe}( \alpha ) \in \Gamma^\omega~|~\alpha \in W\}$. 

For the monitored architecture $(\calM \times \hat{\calP}^1, \hat{\calP}^2, \dots, \hat{\calP}^n, 
\hat{\calM})$ and the specification~$\hat{\Phi}$,  
every correct distributed program corresponds to a finite-state distributed 
winning strategy in the original game, and vice versa.

To conclude the reduction, we restrict the specification~$\hat{\Phi}$ 
to the set of global runs generated by the processes of this architecture, 
and replace the processes with black boxes, thus obtaining  
a hard-wired pipeline architecture in the framework of
\cite{PnueliRos90}, which preserves the solutions 
of the distributed synthesis problem for the game~$G$ with hierarchical observation 
from the outset.
\qed
\end{proof}

Interestingly, the direct translation of pipeline architectures into games
does not necessarily result in games with hierarchical observation. 
Consider, for instance a pipeline with three processes~$1 \le 2 \le 3$, each with a one-bit 
communication alphabet. In the corresponding game, while player receives an input bit from the 
environment, player~$2$ can play an action that 
reveals a bit to process~$3$ which is not revealed to player~$1$, thus the information sets of 
player~$1$ and player~$3$ become incomparable. Indeed, the 
architecture model 
does not a priori prevent processes to emit signals that are independent of the received input. 
The condition is ensured only in the implemented system, 
where every process follows its 
prescribed program. For arbitrary architectures (already in the Pnueli-Rosner framework), it  
seems hard to incorporate this condition 
into the process of designing a correct distributed program.

For the case of pipelines, it is easy to work around this circularity. 
The idea is to send the input of every player to all the previous players, as illustrated in 
Figure~\ref{sfig:pipeline}. 
Formally, this amounts to transforming a given architecture~$\calA$ 
on the communication graph 
$(\{1, \dots, n\}, \to)$ 
by adding \emph{feedback links} $i \to j$ 
from any process $i$ to all processes $j < i$. 
Clearly, the resulting architecture~$\calA'$ corresponds to a game with hierarchical observation. 
We argue that the distributed synthesis problem is invariant under this 
 transformation

\begin{figure}
\begin{center}
\subfigure[pipeline with added feedback links]{
  \label{sfig:pipeline}
  \begin{tikzpicture}[scale=0.8]
    
    \node[state1P,minimum width=.5cm,minimum height=.5cm] (p1) at (0,4) {$1$};
    \draw[->,>=latex]  (-1,4) --  (p1.west);
    \node[state1P,minimum width=.5cm,minimum height=.5cm] (p2) at (1.5,4) {$2$};
    \draw[->,>=latex] (p1) -- (p2);

    \node[state1P,minimum width=.5cm,minimum height=.5cm] (p3) at (3,4) {$3$};
    \draw[->,>=latex] (p2) -- (p3);
    
    \node[state1P,minimum width=.5cm,minimum height=.5cm] (p4) at (4.5,4) {$4$};
    \draw[->,>=latex] (p3) -- (p4);

    \node[state1P,minimum width=.5cm,minimum height=.5cm, draw=none, fill=none] (p5) at (6,4) {~};
    \draw[->,>=latex] (p4) -- (p5);
    
    \draw[->,>=latex, dotted] (p4) edge[bend right=30] (p1);
    \draw[->,>=latex, dotted] (p4) edge[bend right=25] (p2);
    \draw[->,>=latex, dotted] (p4) edge[bend right=15] (p3);
    \draw[->,>=latex, dotted] (p3) edge[bend right=25] (p1);
    \draw[->,>=latex, dotted] (p3) edge[bend right=15] (p2);
    \draw[->,>=latex, dotted] (p2) edge[bend right=15] (p1);
    
  \end{tikzpicture}
}
\hspace*{.3cm}
\subfigure[two-way chain]{ 
  \begin{tikzpicture}[scale=0.8]
      
        \node[state1P,minimum width=.5cm,minimum height=.5cm] (p1) at (0,4) {$1$};
        \draw[->,>=latex]  (-1,4) --  (p1.west);
         \node[state1P,minimum width=.5cm,minimum height=.5cm] (p2) at (1.5,4) {$2$};
         \draw[->,>=latex] (p1.15) -- (p2.165);
         \draw[->,>=latex] (p2.195) -- (p1.345);

          \node[state1P,minimum width=.5cm,minimum height=.5cm] (p3) at (3,4) {$3$};
	  \draw[->,>=latex] (p2.15) -- (p3.165);
          \draw[->,>=latex] (p3.195) -- (p2.345);
          
          \node[state1P,minimum width=.5cm,minimum height=.5cm] (p4) at (4.5,4) {$4$};
          \draw[->,>=latex] (p3.15) -- (p4.165);
	  \draw[->,>=latex] (p4.195) -- (p3.345);
                    
      \end{tikzpicture}
  \label{fig:two-way-chain}
}
 \end{center}
\caption{Pipelines and chains (original links solid, added feedback links dotted)}
\end{figure}
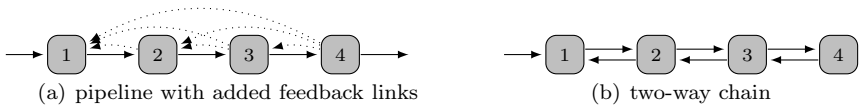 

\begin{lemma}\label{lem:as-if}
Every pipeline can be reduced, by adding feedback links, 
to an architecture that 
corresponds to a game with hierarchical observation and admits the same 
solutions to the distributed synthesis problem.
\end{lemma}

\begin{proof}
Any distributed program for a pipeline architecture~$\calA$ 
generates the same run in the architecture~$\calA'$ with added feedback links 
(which are ignored), hence every solution for~$\calA$ is also 
a correct distributed program for~$\calA'$. 
Conversely, given a distributed program~$\calS'$ for an architecture~$\calA'$ with  
feedback links added to a pipeline architecture~$\calA$ as above, 
we can construct a distributed program~$\calS$ for~$\calA$ by considering, for every 
process~$i < n$, the synchronised product of $\calS'^i$ with $\calS'^{i+1} \times \dots \times 
\calS'^{n}$; hence, 
the current observation of each player $j \ge i$ is maintained in the control 
state. The program $\calS^i$ is built from the product automaton,
by using the observation data from the state rather than the signals 
from the incoming feedback links. 
Accordingly,~$\calS$ is a distributed program for the pipeline~$\calA$ that 
generates the same runs as $\calS'$, hence the transformation preserves 
correctness under any specification.
\qed 
\end{proof}

In terms of games, the argument of Lemma~\ref{lem:as-if}
can be rephrased as follows: in the game corresponding to a 
pipeline architecture, 
the game graph induced by any finite-state strategy profile yields 
hierarchical information.   
Then, we view the output of a hypothetical finite-state 
strategy (program) of every player~$i$ 
as a finite-state signal. 
Due to the pipeline structure, 
this signal is information-consistent for each receiving player~$j \le i$, 
hence it can be made observable (across the feedback links). 
Thus, we obtain a game with hierarchical observation that is 
equivalent to the one that corresponds to the pipeline at the outset.

Using this idea, 
we can recover further results on decidable architectures presented by 
\cite{KupfermanVar01} and \cite{FinkbeinerSch05}. 
Two-way chains, for instance, that is, pipelines 
where every process~$i>1$ 
has an additional link to process~$i-1$ as pictured in 
Figure~\ref{fig:two-way-chain}, 
lead to the same game as the underlying pipeline 
under the transformation of Lemma~\ref{lem:as-if}. 
The case of rings with up to four processes is similar:
a ring is a two-way pipeline with an additional two-way link 
between the first and the last process. 
In a four-process ring as pictured in Figure~\ref{sfig:two-way-ring-four}, for any distributed program, 
process~$1$ can infer the signals emitted by~$3$ (for~$2$ and~$4$), and process $2$ the 
signals emitted by $4$.
Hence, we can add feedback links from $3$ to $1$ and from $4$ to $2$ without changing the 
set of runs generated by distributed programs. By doing so, we obtain a game with
hierarchical observation in the order $1 \preceq 2 \approx 4 \preceq 3$. 
 
However, two-way rings with five or more black-box processes can in general not be transformed 
into games with hierarchical observation by adding feedback links. 
For any linear ordering~$\le$ of the processes, the addition of feedback links from 
each process $i$ to all $j \le i$ 
leads to an architecture that allows spurious runs, 
which cannot be generated by distributed programs 
in the original architecture. Figure~\ref{sfig:two-way-ring-four} illustrates 
that hierarchical observation cannot be attained with the 
ordering $1 \le 2 \le 3 \le 4 \le 5$, for instance, as the feedback link from $5$ to $3$ (dashed in the picture) 
represents a signal that is not observation-consistent for process~$3$.
Indeed, as pointed out by \cite{MohalikWal2003} and \cite{FinkbeinerSch05}, 
the synthesis problem for two-way rings with at least five black-box processes is 
undecidable.

\begin{figure}

\begin{center}
\subfigure[four processes attain hierarchical observation]{
\label{sfig:two-way-ring-four}  
\begin{tikzpicture}[scale=0.7]

\node[state1P] 				(0) at (3,8)         { $1$ };

 \node[state1P]              (00) at (5,10) { $2$ };
 \node[state1P]              (01) at (5,6) { $4$ };

  \node[state1P]              (000) at (7,8) { $3$};

 
 \draw[->,>=latex]  (1.5,8) --  (0.west);
 \path (0) edge 
 (00) edge 
 (01);
 \path (00.west) edge 
 (0.north);
 \path (00) edge 
 (000);
 
 \path (000.north) edge 
 (00.east);
 \path (000) edge  
 (01) ;
\path (01.west) edge 
(0.south); 
\path (01.east) edge 
(000.south); 
\path (000.west) edge[dotted] (0.east);
\path (01.north) edge[dotted] (00.south);
 

\end{tikzpicture}
\hspace*{.5cm}

}
\hspace*{.5cm}
\subfigure[link from $5$ to $3$ (and from $4$ to $2$) is not observation-consistent]{
\label{sfig:two-way-ring-five}  
\begin{tikzpicture}[scale=0.65]

\node[state1P] 				(0) at (3,8)         { $1$ };

 \node[state1P]              (00) at (5,10) { $2$ };
 \node[state1P]              (01) at (5,6) { $5$ };

 \node[state1P]              (000) at (8,10) { $3$};
 \node[state1P]              (001) at (8,6) { $4$};        
        
 
 \draw[->,>=latex]  (1.5,8) --  (0.west);
 \path (0) edge 
 (00) edge 
 (01);

 \path (00.west) edge 
 (0.north);

\path (01.west) edge 
(0.south); 

\path   (000.165) edge (00.15);
\path (00.-15) edge (000.195) ;

\path (01.15) edge (001.165);
\path (001.195) edge (01.-15);

\path (000.285) edge (001.75);
\path (001.105) edge (000.255);
 


\path   (000) edge[dotted] (0);
\path   (001) edge[dotted] (0);
\path   (001) edge[dotted] (00);
\path   (01) edge[dotted] (00);
\path   (01) edge[dashed,color=gray] (000);
 
\end{tikzpicture}
\hspace*{.5cm}
}
\end{center}
\caption{Two-way rings with information-consistent feedback links (dotted)}
\label{fig:rings}
\end{figure}
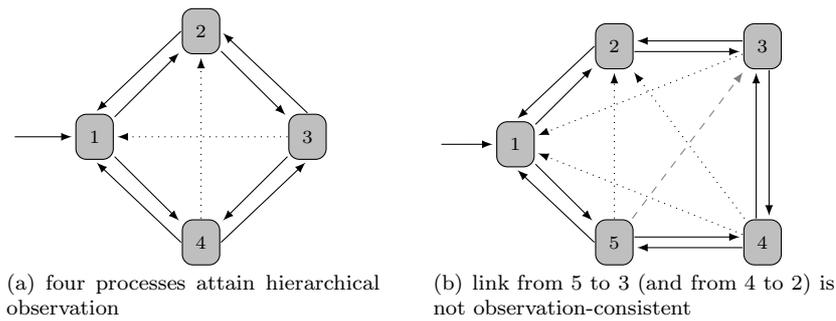

In general, achitectures in the classical framework of Pnueli and Rosner, 
where the communication links are hard-wired, present the following dichotomy: 
either there exists a total ordering $\preceq$ among processes, 
such that the addition of feedback links $\{ i \to j~|~j < i\}$ 
leaves the set of runs generated by any distributed program unchanged, 
or the architecture contains an information fork, in the sense defined 
by~\cite{FinkbeinerSch05} -- in the latter case, the authors of the cited paper show 
that there always exists a specification under which the distributed synthesis problem 
is undecidable.




\section{Effective synthesis for monitored architectures}

By relaxing the condition of hierarchical information 
to include changing or intermittent hierarchies, 
we obtained more general classes of games on which 
the distributed synthesis problem is effectively solvable.
How does this generalisation carry over to distributed architectures\,? 
In the standard framework of architectures with hard-wired communication links, 
there is little hope for redrawing 
the decidability frontier for 
distributed synthesis by 
exploiting patterns of 
dynamic or recurring hierarchical information:
According to the information-fork criterion of~\cite{FinkbeinerSch05}, 
any solvability condition on architectures must either 
restrict the specifications, or request that the 
communication graph is essentially a pipeline. 
Since hierarchical information ---\,in all considered variants\,---  
is a property of the game graph, independent of the winning condition, 
it would be unnatural to cast it as a
restriction on the system specification. 
On the other hand, solvability of pipeline architectures 
is already covered by the basic condition of 
hierarchical observation.

Nevertheless, we encounter situations in practice where 
the components of a distributed system coordinate successfully 
without being restricted to communicate along the links of
a fixed pipeline,
or even a fixed communication graph. 
Consider, for instance, a system that operates
in time phases, each with a specific workflow among
component processes that are organised in a pipeline, 
but just for the duration of one phase; in the next phase, 
the workflow may change and follow another pipeline. 
If, at the end of each phase, all players are updated with the 
information received in the global system, 
we obtain a situation corresponding to a game with 
dynamic hierarchical information. 
Even if the architecture allows any two processes to communicate 
in any direction, the described phase design
ensures that dynamic hierarchical information 
is maintained, hence the synthesis problem is solvable. 
Our shadow-player construction can be understood as an instance 
of this idea.
Beyond maintaining reconfigurable pipelines, one may furthermore 
allow workflows that propagate incomparable information, 
as long as this occurs for a bounded number of rounds.  

\subsection{Hierarchical architectures}

One way to ensure that the synthesis problem is solvable on a class of instances
is by restricting the architecture so that the corresponding game graph yields (static, dynamic, or recurring) 
hierarchical information. 
To verify whether a monitored architecture satisfies the condition, 
we can use the procedures from Lemma~\ref{lem:deciding-hierarchical-dynamic} and 
Theorem~\ref{thm:deciding-hierarchical-recurring} directly, 
without constructing the corresponding game.  
Given an architecture 
$(\calP^1, \dots, \calP^n, \calM)$ 
the procedure for dynamic (or static) hierarchical information 
runs in space $O( \sum_{i=1}^n \log(|\calP^i|) + \log( |\calM | ))$, and
for the case of recurring hierarchical information, 
it runs in space $O( (n | \bigtimes_{i=1}^n \calP^i \times \calM |)^2 )$.

In the fixed-architecture framework, solvable classes are characterised by a basic pattern 
of the communication graph: a pipeline where
no process $i$ is allowed to receive signals from any process $j < i - 1$. 
This ensures that every run maintains hierarchical observation, hence the  
game graph corresponding to the architecture yields hierarchical observation, which implies that  
the synthesis problem is solvable.
We can identify similar solvability patterns for monitored architectures by restricting to
view monitors that
never deliver signals
which would violate the condition of dynamic hierarchical information. 
For any architecture equipped with such a monitor, the synthesis problem 
is solvable with respect to every specification.
Deciding whether a given view monitor~$\calM$ yields dynamic or recurring hierarchical information 
(in the corresponding game graph) 
with every matching collection of processes amounts to deciding whether 
the architecture formed of black-box processes and the
view monitor~$\calM$ yields dynamic or recurring hierarchical information, and can hence
be done in logarithmic or polynomial space, respectively.  

Thus, we obtain classes of monitored architectures on which the 
distributed synthesis problem is solvable as a direct application of our 
game-theoretic analysis.

\subsection{Maintaining hierarchical information through strategies}\label{ssec:strategy-maintained}

Our second proposal for automated 
synthesis in the framework of monitored architectures
relies on enforcing hierarchical information
strategically, at the program level, 
rather than restricting the architecture a-priori.
In this way,
the task of avoiding incomparable information is put in the hands of the 
designer of the distributed system.   
We describe an automated 
method to help 
the system designer accomplish this task.

Let us first detail our argument in terms of games.
Note that, in view of recurring hierarchical information, 
it is undecidable whether, for a given game, there exists a strategy  
that is winning and also
avoids infinite gaps with incomparable information:
this follows by adapting the 
standard reduction from the Halting Problem to the synthesis  
of reachability strategies in two-player games with imperfect information 
(see, e.g., \cite{PetersonRei79}, \cite{BerwangerKai10}).
Therefore, we restrict our attention to dynamic hierarchical information.

\begin{definition}\label{def:hierarchic-info-strat}
  Given a game, a strategy~$s$ maintains
  hierarchical information 
  if every history that follows~$s$ yields hierarchical information.
\end{definition}

One straighforward, but important insight is that 
the synthesis problem 
restricted to strategies that maintain hierarchical information 
is effectively solvable. 
 
\begin{theorem}\label{thm:hierarchy-synthesis}
For any finite game, it is decidable whether  
there exists a distributed winning strategy that maintains hierarchical 
information, and if so, we can synthesise one.  
\end{theorem}

\begin{proof}
Let~$\calG$ be an arbitrary finite game with an $\omega$-regular 
winning condition. 
According to Lemma~\ref{lem:automaton-hierarchical-history},
the set of game histories that yield hierarchical information is regular.
Let~$A$ be a deterministic automaton that recognises this set 
and consider its synchronised product $G \times A$ with the game graph~$G$.
From this product, we construct a new game graph~$G'$ by adding a
sink position~$\ominus$ with a fresh observation to be received by all players, 
and by replacing all moves $((v, q), a, (v, q'))$ in $G \times A$ 
where the automaton state~$q'$ at the target is non-accepting 
with $((v, q), a, \ominus)$. Hence, all histories in~$G'$ 
yield hierarchical information and
every history in $G$ that does not yield hierarchical information,  
maps to one in~$G'$ that ends at~$\ominus$. 
Finally, we
adjust the winning condition~of $\calG$ by expanding each play with 
the corresponding run of~$A$ 
and by excluding 
all plays that reach~$\ominus$.  

The game~$\calG'$ constructed in this way yields 
dynamic hierarchical information, 
hence the synthesis problem is effectively solvable.  
Moreover, every distributed winning strategy in~$\calG'$ corresponds 
to a distributed strategy 
in~$\calG$ that is winning and maintains hierarchical information
in the sense of Definition~\ref{def:hierarchic-info-strat}, and
conversely, every winning strategy that maintains hierarchical information in 
$\calG$ corresponds to a winning strategy in~$\calG'$. 
\qed
\end{proof}

The above theorem gives raise to an effective, sound, and incomplete method for 
solving the synthesis problem for monitored architectures: 
Given a problem instance consisting of an architecture and a specification, 
solve the synthesis problem for the corresponding game 
restricted to strategies that maintain hierarchical information, and translate 
any resulting finite-state winning strategy back into a distributed program.
This approach omits 
solutions that involve runs which do not yield hierarchical information. However,   
if there exist 
correct distributed programs that also maintain hierarchical information, 
the procedure will construct one.

Alternatively, the method based on Theorem~\ref{thm:hierarchy-synthesis} 
can be viewed as a complete procedure for solving the synthesis problem on 
instances with arbitrary architectures, 
but with specifications restricted to run trees in which every run (corresponds 
to a play that) yields dynamic hierarchical information -- a regular property 
according to Lemma~\ref{lem:automaton-hierarchical-history}. 

\subsection{Hierarchical routing}\label{ssec:routing}

To conclude, we present a concrete example of an architecture design 
that supports the application of our method without 
otherwise restricting the solution procedure.

We set out with the observation that
the condition of dynamic (or static) hierarchical information is a safety condition that 
can be supervised by the view monitor. 
Our proposal is to incorporate into view monitors the ability to send hierarchy-related 
data to the processes, which can be used by the programs
to ensure that hierarchical information is maintained whenever possible.

To illustrate the idea, we consider architectures of a particular format, 
which we call \emph{routed} architectures. Intuitively each processes 
can emit signals addressed to any other process, and the view monitor, called \emph{router}, 
either delivers a signal or denies it, according to a deterministic rule. 
In case of denial, the sending process receives a notification. The intention is to 
maintain hierarchical information on every run as far as possible. However, 
signals sent by the Environment cannot be denied, 
and inter-process signals may also be forced for delivery by 
the emitting process. 
This may lead to violations of the condition of hierarchical information, 
in which case the monitor sends a panic signal to all processes 
and henceforth simply delivers all emitted signals. 

Formally, a routed architecture for $n$ processes and the Environment features 
action of a special format. Each action has a control and a communication component:  
the control component is a single symbol as in the case of static architectures; 
the communication component is a (possibly empty) list of signals, 
each formed of a body -- a single symbol -- and a header, which contains the identities $i, j$ of the sending and the 
receiving process, and a priority number. 
In Environment actions, all signals carry top priority (which forces delivery). 
We assume that no process appears twice as a receiver in one action. 
The observations of each process are formed of a communication component, 
which consists of a list of signals, at most one from every other process, 
and a notification component, which consists of 
a panic flag and a list of delivery flags, one for every other process. 

Priority numbers are used to determine the observations delivered in response 
to a global action. Towards this, 
the priorities of all signals in the communication component of an observation are aggregated. 
The aggregation function is monotonous with respect to the inclusion between sets of 
signals, and each set of observations has a unique element of maximum 
value (to break ties, we may use process identifiers). 
Since the observation space is finite, it is always possible to define a priority aggregation function 
with these properties.     

The \emph{router} for a given collection of processes
is a view monitor that operates as follows. There is one sink state called panic state: 
here the router reads the 
communication components of the global action and delivers to each process~$i$ the list of all signals 
addressed to $i$ as a receiver, also setting the panic flag and 
delivery flags for each signal emmited by process~$i$. 
In any other states, including the initial one,
the router reads the global action~$a$ and considers the set of
admissible observations: a 
global observation~$b$ is admissible if (1) the communication component consists of signals emitted in the
global action~$a$ and contains all signals with top priority, and (2) 
the delivery flags are set correctly, and the panic flag is reset, if and only if, 
the run prefix that would be reached by delivering the observation $b$ (corresponds to a play that) 
yields hierarchical information. 
If there exist admissible observations that do not raise the panic flag, 
the monitor picks the one of maximal aggregated priority, delivers it to the processes, and 
switches into a non-panic successor state. Else, if all admissible observations raise the panic flag, 
the monitor delivers all received signals and switches into the panic state. 
As the condition of hierarchical observation corresponds to a finite-state signal, 
the described operation of the router can be implemented by a finite-state monitor. 
Essentially, the states store the data needed to supervise the 
information hierarchy ordering. 

In terms of expressiveness, routed architectures lie between hard-wired and monitored architectures. 
As in hard-wired architectures, signals are routed unaltered between processes, 
the difference being that contents and receiver of a signal are chosen by the 
emitting process. In contrast, the signals delivered by an arbitrary view monitor 
to different processes can be highly correlated depending on
on the global action and the state of the monitor. 
For routed architectures, the correlation is restricted to 
non-delivery of signals due to violating the condition of 
hierarchical information. 

Given a routed architecture and a (linear or branching-time) specification, we can use the incomplete 
synthesis method presented in Subsection~\ref{ssec:strategy-maintained} to 
synthesise a distributed program that maintains hierarchical information. If this succeeds, 
we obtain a correct distributed program for which no panic or non-delivery flag will ever be raised.

However, assuming that the given specification is insensitive to signals that are not delivered, 
it is sufficient to synthesise a distributed strategy 
that avoids the panic signal, but allows non-delivery of signals.
In game-theoretic terms, this corresponds to viewing signalling attempts as cheap-talk actions. 
The option of sending signals that may be denied does 
not enlarge the class of solvable problem instances:  in principle, each process~$i$ can infer 
from its observation sequence that a certain signal will not be delivered, since this
can occur only when the receiving process is less informed than~$i$ -- in other words, the non-delivery 
notification is information-consistent for process~$i$, so it yields no new information. 
Nevertheless, in practice, the warning mechanism against driving the system into 
a non-hierarchical state offers the system designer more freedom in programming the processes 
than the setting where such attempts are forbidden.

Finally, for problem instances that do not admit a complete hierarchical solution, 
our model of routed architectures allows to combine strategies synthesised 
via the method from Subsection~\ref{ssec:strategy-maintained} 
applied on the prefix of the run tree on which hierarchical information is maintained   
with other strategies that are triggerred when the panic signal is raised. 

\section{Discussion}

The task of coordinating several players with imperfect information 
towards attaining a common objective is quintessential in the design of
computational systems with multiple interacting components. 
Known automated methods for accomplishing this task rely on a basic pattern of 
hierarchical information flow. 
Here, we showed that this pattern can be relaxed considerably, 
towards allowing the hierarchy to 
rearrange in the course of the interaction, or even to vanish temporarily.

Our technical analysis is based on finite-state games with 
perfect recall and synchronous dynamics. The model has the advantage of 
exposing the fundamental connection between knowledge and action
as put forward by~\cite{Moses16}. Indeed, our results on 
solvable cases for the coordination problem in games rely on 
the key concept of identifying 
finite-state signals that provide sufficient knowledge 
for triggering winning strategies. 

The game model has shortcomings as well.  
The assumption of perfect recall is rather uncritical, 
as we finally synthesise strategies implementable by 
finite-state automata. However, the assumption of synchronous dynamics does restrict 
the scope of our results: 
the uncertainty about the ordering of event occurrences in asynchronous systems 
cannot be captured directly by imperfect information in games, 
and different methods are required to approach the synthesis problem.
(See, e.g., ~\cite{MadhuThiagu02}, \cite{GastinSZ09}, \cite{MuschollWal14}.)
Insights on information and coordination in games
may shed light on the asynchronous setting as well, however, at the current state of research 
the areas appear divided. 

Perhaps the greatest challenge is to make concrete how the game-theoretic results 
can help in designing real-world systems.
In terms of operational models, 
our games are close to the distributed reactive systems of \cite{PnueliRos90}, 
described by a communication architecture and a regular specification:
Instances of the synthesis problem can be translated back and forth between the game and the 
system model. 
Nevertheless, it turns out that games with (static, dynamic, or recurring) hierarchial information,  
do not correspond to natural classes of distributed architectures. 
This is because most of the interaction structure described by a game graph 
is incorporated into the specification on the side of the architecture. 
While the game graph expresses which actions \emph{will be} available or not in a certain state, 
the specification expresses which actions \emph{should}  occur or not in a run of a correct program. 
Since the original model of distributed reactive systems is too permissive in describing
possible behaviours of processes,  a meaningful classification 
in terms of information flow patterns 
would need to refer to behaviours of correct programs -- which we aim yet to construct.

To overcome the cleavage between game and architecture representations, 
we introduced the model of monitored architectures where
the hard-wired communication graph of the Pnueli-Rosner model is replaced by a transducer 
that transforms global actions into signals sent to the processes, in a dynamic way, depending on the previous 
run. This yields a faithful representation of the information flow
in the operational model, which allows to translate 
the decidability results on games with hierarchical information. Thus, we obtain
rich classes of distributed architectures 
on which the automated synthesis problem is solvable. 

In continuation of our study, we see three directions for further research. 
One promising approach is to refine the classification of solvable classes according to
the structure of the winning condition (or the specification) 
rather than considering only the set of possible behaviour represented by the game 
graph (or the architecture). In the classical framework of Pnueli and Rosner,
\cite{MadhusudanT01} considered specifications that are \emph{local}, in the sense that they require each 
process to satisfy a condition expressed on the state space of its own automaton. The authors show that, 
under such specification, the class of solvable architectures includes clean pipelines, 
that is pipelines where the last process receives private signals from the environment. Technically, 
the games corresponding to clean pipelines do not yield hierarchical information, nevertheless, the synthesis 
problem can be solved by a straightforward adaptation 
of the automata-theoretic method of~\cite{PnueliRos90,KupfermanVar01}. 
As another example, our method for synthesising distributed programs that maintain hierarchical information, 
presented in Subsection~\label{ssec:strategy-maintained}, can be interpreted as a complete solution for 
arbitrary architectures with specifications restricted to runs that 
yield dynamic hierarchical information. 

A second concrete direction relies on  
assessing the condition of hierarchical information 
at runtime rather than considering the raw system model.
In terms of games, this corresponds to asserting that 
the game graph induced by every strategy yields hierarchical information, 
although this may not be true for the original game.
Such an approach would allow to capture functional dependencies in runs generated by distributed programs, 
and lead to even more liberal classes of solvable architectures.

Finally, we may include the view monitor in the synthesis process. 
So far, we modeled the view monitor as a white-box process that represent the available communication infrastructure. 
However, in practical applications, the infrastructure does not need to be fixed, 
the system designers may be able to configure certain of its parameters. 
On the other hand, as we showed in Subsection~\ref{ssec:routing}, 
it can be helpful to use the abilities of the view monitor as a global observer 
to provide the processes with signals about the global run, even if they 
are observation-consistent and hence could be deduced locally by the receiver. 
Therefore, it will seems appropriate to consider 
view monitors as grey-box processes, and synthesise white-box implementations using automated procedures.

At bottomline, we are confident that the quest for 
effective automated synthesis is worth pursuing. 
To correct the discouraging picture drawn by the fundamental undecidability results, it is important to 
identify natural classes of models on which the problem is solvable. We believe that hierarchical 
information offers a convincing example in this direction. 

\bibliographystyle{my-spr-chicago}      
\bibliography{hi}   

\end{document}